\documentclass[preprintnumbers,floatfix,superscriptaddress,nofootinbib,prx,twocolumn,showpacs,groupedaddress]{revtex4-1}
\usepackage[T1]{fontenc}
\usepackage{amsmath,amsthm,amssymb}
\usepackage{graphicx}
\graphicspath{{Figures/}}
\usepackage{dcolumn}
\usepackage{bbold}
\usepackage{hyperref} 
\usepackage{braket}
\usepackage{xcolor}
\usepackage{dsfont}
\usepackage{setspace}
\usepackage{multirow}
\usepackage{booktabs}
\usepackage[justification=raggedright]{caption}
\usepackage{lipsum}
\usepackage{algorithm}
\usepackage{physics}

\newcommand{\tn}{\textnormal}
\newcommand{\inn}{\textnormal{in}}
\newcommand{\out}{\textnormal{out}}
\newcommand{\guess}{\textnormal{guess}}

\newcommand{\vl}[1]{\textcolor{black}{#1}}
\newcommand{\change}[1]{\textcolor{black}{#1}}
\newcommand{\de}[1]{\left(#1\right)}

\newcommand{\DE}[1]{\left\{#1\right\}}
\newcommand{\hon}{\mathcal{H}}
\newcommand{\adv}{\mathcal{A}}

\theoremstyle{definition}

\newtheorem{definition}{Definition}
\newtheorem{assumption}{Assumption}

\theoremstyle{plain}
\newtheorem{lemma}{Lemma}
\newtheorem{theorem}{Theorem}

\usepackage{array}
\newcolumntype{L}[1]{>{\raggedright\let\newline\\\arraybackslash\hspace{0pt}}m{#1}}
\newcolumntype{C}[1]{>{\centering\let\newline\\\arraybackslash\hspace{0pt}}m{#1}}
\newcolumntype{R}[1]{>{\raggedleft\let\newline\\\arraybackslash\hspace{0pt}}m{#1}}

\begin{document}

\title{Anonymous transmission in a noisy quantum network using the $W$ state}
\author{Victoria Lipinska}\email{v.lipinska@tudelft.nl}

\author{Gl\'{a}ucia Murta}\email{g.murtaguimaraes@tudelft.nl}

\author{Stephanie Wehner}

\affiliation{QuTech, Delft University of Technology, Lorentzweg 1, 2628 CJ Delft, The Netherlands}

\date{November 23, 2018}

\begin{abstract}
We consider the task of anonymously transmitting a quantum message in a network. We present a protocol that accomplishes this task using the W state and we analyze its performance in a quantum network where some form of noise is present.
We then compare the performance of our protocol with some of the existing protocols developed for the task of anonymous transmission. We show that, in many regimes, our protocol tolerates more noise and achieves higher fidelities of the transmitted quantum message than the other ones. Furthermore, we demonstrate that our protocol tolerates one nonresponsive node. We prove the security of our protocol in a semiactive adversary scenario, meaning that we consider an active adversary and a trusted source.
\end{abstract}

\maketitle

\section{Introduction}
In cryptographic scenarios we are often concerned with hiding the content of the messages being exchanged. However, sometimes the identity of the parties who communicate may also carry relevant information. Examples of tasks where the identities of the ones who communicate carry crucial information are voting, electronic auctions \cite{Stajano2000} or, more practically, sending a message to a secret beloved \cite{Chaum1981}. Therefore, the establishment of anonymous links in a network, where identities of connected parties remain secret, is an important primitive for both classical \cite{Chaum1988} and quantum communication.

In this paper we consider a task of anonymously transmitting a quantum message in a network.
To define the task more precisely, consider a quantum network with $N$ nodes. One of the nodes, sender $S$, would like to communicate a quantum state $\ket{\psi}$ to a receiver $R$ in a way that their identities remain completely hidden throughout the protocol. In particular, for $S$ it implies that her identity remains unknown to all the other parties, whereas for $R$ it implies that no one except $S$ knows her identity. The essence of the protocol is to create an entangled link between $S$ and $R$ by performing local operations on the other nodes of the network. Such a link is called \textit{anonymous entanglement} (AE) \citep{Christandl2005}, since the identities of the nodes holding the shares of the entangled pair is kept anonymous. After  anonymous entanglement is created, $S$ and $R$ use it as a resource for teleporting the quantum information $\ket{\psi}$. Note that the main goal of anonymous transmission is to fully hide the identities of the sender and the receiver; it does not aim at guaranteeing the reliability of the transmitted message.

A number of protocols have been proposed to tackle this task, which was first introduced in \cite{Christandl2005}. There, the authors present a protocol which makes use of a given multipartite Greenberger-Horne-Zeilinger (GHZ) state as a quantum resource, i.e., {$\ket{\text{GHZ}_N} = \frac{1}{\sqrt{2}}(\ket{0\dots 0} + \ket{1 \dots 1})$}. The problem was subsequently developed to consider the preparation and certification of the GHZ state \cite{Bouda07, Brassard2007}. In \cite{Brassard2007}, it was first shown that the proposed protocol is information-theoretically secure against an active adversary. What is more, other protocols were proposed, which do not make use of multipartite entanglement, but utilize solely Bell pairs to create anonymous entanglement \cite{Yang2016}. Yet, so far, it has not been discussed whether multipartite states other than the GHZ allow for anonymous transmission of a quantum state. Moreover, nothing is known about the performance of such protocols in a realistic quantum network, where one inevitably encounters different forms of noise.

Here we design a protocol for quantum anonymous transmissions which uses the W state, {$\ket{\tn{W}}_N = \frac{1}{\sqrt{N}}(\ket{10\dots0} + \dots + \ket{0\dots01})$}. Just like other existing protocols, our protocol is based on establishing anonymous entanglement between $S$ and $R$. We prove the security of our protocol in a semiactive adversary scenario, meaning that we consider an active adversary and a trusted source, as in \citep{Christandl2005}.
We also show that security is preserved in the presence of noise in the network, when all the particles are subjected to the same type of noise. What is more, we compare the performance of our protocol with previously proposed protocols that use the GHZ state and Bell pairs. We quantify the performance of protocols by the fidelity of the transmitted quantum state. We find that, in many cases, our W-state  based protocol tolerates more noise than the other protocols and achieves higher fidelity of the transmitted state. Additionally, we show that our protocol can tolerate one nonresponsive node, e.g., if one of the qubits of a multipartite state gets lost.  In contrast, the protocol using the GHZ state cannot be carried out at all in this case, since the loss of a single qubit destroys the entanglement of the state. We also address the performance of the Bell-pair based protocol, presented in \cite{Yang2016}, and we show that in the presence of noise, the performance of the protocol depends on the ordering of $S$ and $R$ in the network. To the best of our knowledge this is the first analysis of anonymous transmission in the presence of noise. Without such an analysis the performance of near-future applications for quantum networks cannot be characterized \cite{Qroadmap}.

{The paper is organized as follows. In Sec.~\ref{sec:protocol}, we present the protocol for anonymous transmission with the W state and discuss its correctness. In Sec.~\ref{sec:security_analysis}, we provide the security definition and prove that our protocol is secure in the semiactive and passive adversary scenario. Finally, in Sec.~\ref{sec:performance} we examine the behavior of our protocol in a noisy quantum network and compare it with the other existing protocols. }

\section{The protocol}\label{sec:protocol}

Our anonymous transmission protocol, Protocol {1}, allows a sender $S$ to transmit an arbitrary quantum state $\ket{\psi}$ to a receiver $R$ in an anonymous way and uses the $N$-partite W state as a quantum resource.
\begin{algorithm}[H]
\caption*{\textbf{Protocol 1:} Anonymous transmission with the W state.}\label{protocol}
\vspace{.5em}
{\textbf{Goal:} Transmit a quantum state $\ket{\psi}$ from the sender $S$ to the receiver $R$, while keeping the identities of $S$ and $R$ anonymous. }\\
\begin{enumerate}
\item \textit{Collision detection.} \\ Nodes run the {classical collision detection protocol \cite{Broadbent2007}} to determine a single sender $S$. All nodes input 1 if they do wish to be the sender and 0 otherwise. If a single node wants to be the sender, continue.
\item \textit{Receiver notification.} \\ Nodes run the {classical receiver notification protocol \cite{Broadbent2007}}, where the receiver $R$ is notified of her role.
\item \label{prot:state_distr} \textit{State distribution.} \\A trusted source distributes the $N$-partite W state. 
\item \textit{Measurement.} \\$N-2$ nodes (all except for $S$ and $R$) measure in the $\{\ket{0},\ket{1} \}$ basis. 
\item \textit{Anonymous announcement of outcomes.} \\Nodes use the 
{classical veto protocol \cite{Broadbent2007}} which outputs 0 if all the $N-2$ measurement outcomes are 0, and 1 otherwise. If the output is 0 then anonymous entanglement is established, else abort.
\item \textit{Teleportation.} \\Sender $S$ teleports the message state $\ket{\psi}$ to the receiver $R$. Classical message $m$ associated with teleportation is sent anonymously. The communication is carried out using {the classical logical OR protocol \cite{Broadbent2007}}
which computes $m \oplus \tn{rand}$, where $\tn{rand}$ is a random 2-bit string input by the receiver $R$.
\end{enumerate}
\end{algorithm}
Protocol {1} is built on a number of classical subroutines -- collision detection, receiver notification, veto and logical OR. {Specifically: collision detection checks whether only one of the nodes wishes to be the sender; receiver notification notifies the receiver of her role in the protocol; veto announces if at least one of the parties has given input 1; 
and logical OR computes the XOR of the input of all the parties.} In Ref. \cite{Broadbent2007}, protocols for implementing these classical subroutines were proposed. The protocols were proven to be information-theoretically secure {in the classical regime}, even with an arbitrary number of corrupted participants, assuming the parties share pairwise authenticated private channels and a broadcast channel. 
However, security against a quantum adversary was not analyzed. Like in related work \citep{Brassard2007}, here we will assume that the protocols listed above remain secure even in the presence of a quantum adversary. We make this assumption explicit in the security proof presented in App. \ref{app:sec:states_and_registers}, where we assume that the classical subprotocols only act on the classical input register and create the output register, therefore, not revealing any information other than what is specified by the protocol.

The main concern of any anonymous transmission protocol is to hide the identities of sender $S$ and receiver $R$. {Nonetheless, it is also desired that, in the case in which all the parties act honestly, no information about the transmitted message is revealed.} In order to achieve this functionality we add the step where $R$ randomizes the output of the logical OR in Step 6 of Protocol {1}. In that way, the classical outcome of the teleportation, $m$, is sent from $S$ to $R$ in a secret way.
Indeed, even though the classical bit $m$ could be sent by a simple anonymous broadcast protocol, the probability of obtaining a particular outcome $m$ can depend on which state is teleported if the established anonymous entanglement is not a maximally entangled state. This is the case especially in the presence of noise in the network (for more details see App. \ref{app:sec:security_analysis}). 

Note that our protocol is probabilistic, as the parties may abort in Step 5. However, since the measurement outcomes are announced, the creation of anonymous entanglement is heralded. Hence, $S$ and $R$ know when the anonymous entanglement failed to be established before they initiate the teleportation, so in the case in which the protocol aborts, $S$ keeps the state $\ket{\psi}$. In the following we first state the correctness of the protocol and then elaborate on the probability of success in the protocol, as a function of the number of parties in the network $N$.

\begin{lemma}[correctness]
	\change{If all the parties act honestly and Protocol {1} does not abort, the state $\ket{\psi}$ is transferred from the sender $S$ to the receiver $R$, except with probability $\epsilon_{\rm corr}$, where $\epsilon_{\rm corr}$ is an exponentially vanishing function of the number of rounds used to implement the classical subroutines}.
\end{lemma}
\begin{proof}
	\change{First, recall that Protocol {1} is built on several classical subroutines and in Ref.~\cite{Broadbent2007}, protocols to implement these subroutines were presented. The protocols were proven to be correct except with a probability that vanishes exponentially with the number of rounds $n_{\rm class}$ used to implement the subroutines. Secondly, conditioned on the fact that the classical subroutines are correct and the parties act honestly, the measurement in the $\{\ket{0},\ket{1} \}$ basis can lead to two situations: \textit{(i)} all parties obtain measurement outcome 0, in which case the anonymous entangled state between $S$ and $R$ is $\ket{\psi^+} = \frac{1}{\sqrt{2}}(\ket{01} + \ket{10})$, or \textit{(ii)} a single party obtains a measurement outcome 1 and then the state between $S$ and $R$ is $\ket{00}$, {in which case they abort the protocol}. If the parties do not abort the protocol in Step 5, then the state shared by $S$ and $R$ is the maximally entangled state $\ket{\psi^+} = \frac{1}{\sqrt{2}}(\ket{01} + \ket{10})$, which is then used to perfectly teleport state $\ket{\psi}$ from $S$ to $R$. Altogether, this implies that Protocol {1} is correct except with probability $\epsilon_{\rm corr}$ which vanishes exponentially with $n_{\rm class}$.}
	
\end{proof}

\begin{lemma}[probability of success]\label{lemmaprob}
\change{Given sender S and receiver R, the probability of obtaining the anonymous entangled state {$\ket{\psi^+}$} in Step 4 of Protocol~{1} is $\frac{2}{N}$.}
\end{lemma}
\begin{proof}
Let $|\vec{0}\rangle\!\langle\vec{0}|_{N-2}$ denote the projection on the $\ket{0}$ state of $N-2$ parties. The probability $P_{\psi^+}$ of obtaining this state can be expressed as $P_{\psi^+} = \Tr[\ketbra{\tn{W}}_N \cdot \left(\mathbb{1}_{SR} \otimes |\vec{0}\rangle\!\langle\vec{0}|_{N-2}\right)] = \frac{2}{N}\Tr[\ketbra{\psi^+}{\psi^+}] = \frac{2}{N}$.
\end{proof}

\change{Lem. \ref{lemmaprob} states that in the honest implementation, the probability of not aborting in Step 4 of Protocol {1} decreases with the number of parties. Protocols based on the GHZ state \cite{Christandl2005,Brassard2007}, on the other hand, are deterministic in creating anonymous entanglement.} However, we remark that a fair comparison between the success rate of the two protocols should also take into account the rate of state generation.
Note that recently, a linear optical setup for generating the W state in {nitrogen-vacancy systems} was proposed \cite{Kalb2018}, which could offer a potential advantage in generation rates of the W state, over the GHZ state.

\section{Security}\label{sec:security_analysis}

As discussed in the previous section, in the task of anonymous transmission the main goal is to keep the identities of sender $S$ and receiver $R$ secret. In this section we present the security definitions and prove the security of Protocol {1} against a semiactive adversary. 

{ Let $[N] = \{1,\dots,N \}$ be the set of nodes. We say that dishonest nodes are a subset $\mathcal{A} \in [N] $, with $|\mathcal{A}| = t$.
{This set is defined at the beginning of the protocol, which  is} known as a \textit{nonadaptive} adversary.}

\begin{definition}[semiactive adversary]
We define the \textit{semiactive adversary} scenario as one in which the adversaries are active, i.e., can perform arbitrary {joint} operations on their state during the execution of the protocol, but the source distributing a quantum state is trusted.
\end{definition}

{In particular, for Protocol 1 this means that the state in Step 3 is exactly the W state}. This adversarial model is stronger than a \textit{passive} adversary, where it is assumed that the parties follow all the steps of the protocol and only collect the available classical information. However, note that a fully active adversarial scenario would allow the cheating participants to corrupt the source.

We define security in terms of the guessing probability, i.e., the maximum probability that adversaries guess the identity of the $S$ or $R$ given all the classical and quantum information they have available at the end of the protocol. {Intuitively, we say that the protocol is secure when the guessing probability is no larger than the uncertainty the adversaries have about the identity of the sender before the protocol begins. This uncertainty is defined by the prior probability, $P[S=i|S\notin \mathcal{A}]$. For example, in the case where all the nodes are equally likely to be the sender, the prior probability is uniform and, therefore, $P[S=i|S\notin \mathcal{A}] = \frac{1}{N-t}$.}

\change{In Protocol {1} it is assumed that the message $\ket{\psi}$ to be sent carries no information about the sender's identity. We remark that anonymous transmission is concerned with ensuring anonymity and not secrecy. In the case in which secrecy of the message is required, anonymous transmission could be combined with another primitive that allows one to encrypt the message. However, here, we do not address this issue.}

\begin{definition}[guessing probability]
Let $\mathcal{A}$ be the subset of semiactive adversaries. Let $C$ be the register that contains all classical and quantum side information accessible to the adversaries. Let $W^\adv$ denote the adversaries' quantum register of the state distributed by the source. Then, the probability of adversaries guessing the sender is given by
\begin{align}\label{EQ:sec_active2}
\begin{split}
& P_{\guess}[S|W^\adv,C, S\notin \mathcal{A}] = \\
& \quad = \max_{\{M^i\}} \sum_{i \in [N]} P[S=i|S\notin \mathcal{A}] \Tr[M^i \cdot \rho_{W^\adv C|S=i} ], 
\end{split}
\end{align}
where the maximization is taken over the set of POVMs ${\{M^i\}}$ for the adversaries and $\rho_{W^\adv C|S=i}$ is the state of the adversaries at the end of the protocol, given that node $i$ is the sender. 
\end{definition}

\begin{definition}[sender security]
We say that an anonymous transmission protocol is \textit{sender-secure} if, given that the sender is honest, the probability of the adversary guessing the sender is 
\begin{align}
P_{\guess}[S|W^\adv,C,S\notin \mathcal{A}] \leq \max_{i\in[N]} P[S=i|S\notin \mathcal{A}].
\end{align}
\end{definition}

In words, the protocol is sender-secure if the probability that the adversaries guess the identity of $S$ at the end of the protocol is not larger than the probability that an honest node $i$ is the sender, maximized over all the nodes.  An analogous definition can be given for the \textit{receiver security}. 

{We remark that, even  if $S$ and $R$ are honest, it is trivially possible for the \textit{malicious} parties to prevent $S$ and $R$ from exchanging the desired message. For example, the dishonest parties can measure the W state in a different basis affecting the resulting anonymous entanglement. In this sense, the correctness of Protocol~{1} is not robust to malicious attacks. However, in what follows, we show that Protocol~{1} is secure, and even in the presence of dishonest parties, the anonymity of $S$ and $R$ is preserved.}

\begin{theorem}\label{thm:WSsecure}
The anonymous transmission protocol with the W state, Protocol 1, is sender- and receiver-secure in the semiactive adversary scenario.
\end{theorem}

\begin{proof}[Idea of the proof.]
For clarity, here we present the main idea of our security proof and we refer the reader to 
App.~\ref{app:sec:security_analysis} for details. Note that in the semiactive adversary scenario we allow the adversaries to apply an arbitrary cheating strategy, which in particular includes not following the steps of the protocol and performing global operations on their joint state. First, let us discuss the sender security. We consider the case when $R$ is honest, $R\notin \adv$, as well as when she is dishonest, $R\in\adv$. In both cases, the gist of our sender-security proof is to show that the reduced quantum state of the adversary $\rho_{W^\adv C|S=i}$ at the end of the protocol is independent of the sender, i.e., $\forall i \notin \adv$, $\rho_{W^\adv C|S=i} = \rho_{W^\adv C}$. To show it, we explicitly use the assumption that the classical protocols do not leak any information about $S$ or $R$'s identity even if the adversary has access to quantum correlations. Therefore, any quantum side information the adversary holds is independent of $S$. This, together with the fact that the state distributed by the source is permutationally invariant yields the desired equality. Since now the reduced quantum state of the adversary is independent of $S$ we can easily upper-bound the guessing probability by $\max_{i \in [N]} P[S=i|S\notin \mathcal{A}]$. {The receiver security can be proven following the same structure.}
\end{proof}
{Note that our security proof tolerates any number of cheating nodes. It is also general enough to make a security statement about any resource state that is invariant under permutation of nodes.}

Let us now discuss a passive adversarial model, also called the honest-but-curious model.
This is the case when the malicious parties follow all the steps of the protocol (in particular, they measure in the $\DE{0,1}$ basis in Step 4), but can collaborate to compare their classical data. Note that the passive adversary model is a special case of the semiactive adversary scenario.
{However, this model is interesting by itself, since in the case in which the nodes build their anonymous transmission protocol using weaker versions of classical subroutines, i.e., those that are not secure against quantum adversary,  the security still holds.}
  Indeed, it restricts the power of the adversary, so that they cannot share any quantum side information.  Then, the probability of the adversaries guessing the sender simplifies to
$
P_{\guess}[S|W^\adv,C, S\notin \mathcal{A}] = \sum_{a,c}P[W^\adv=a,C=c] \max_{i \in [N]} P[S=i|W^\adv=a,C=c,S\notin \mathcal{A}], 
$
where maximization is taken over all the values of the random variable $S$, and $a,c$ are possible values of random variables $W^\adv$ and $C$ respectively \cite{Tomamichel2012}. Note that, unlike before, here $W^\adv$ is a classical register of the adversary, since their share of the W state was measured in the $\DE{0,1}$ basis. An analogous expression holds for receiver-security. 

\begin{theorem}\label{thm:passive_security}
The anonymous transmission protocol with the W state, Protocol 1, is sender- and receiver-secure in the passive adversary scenario.
\end{theorem}
The proof of this statement is a special case of the proof of Thm.~\ref{thm:WSsecure}. As before, we use the fact that classical protocols do not leak identities of $S$ and $R$ and the permutational invariance of the resource state {to conclude that the classical information generated during the protocol is independent of who is sender and receiver}. For details see App.~\ref{app:sec:security_analysis}.

\section{Anonymous transmission in a noisy quantum network}\label{sec:performance}
Equipped with the security tools from the previous section, here we analyze the security and performance of Protocol~{1} in a noisy quantum network. We consider a noise model in which each qubit is subjected to the same individual noisy channel. One can think that a trusted source prepared the multipartite state for the network, but each qubit is individually affected by a noise map $\Lambda$ while being transmitted to the nodes. Note that this model can also encompass noise on the local measurements performed on the state. Therefore, in our noisy network,
if $\ket{\tn{W}}_N$ is the perfect $N$-partite W state prepared by a trusted source, then
\begin{align} \label{eq:n-fold_noise}
\omega^\Lambda_N = \Lambda^{\otimes N} (\ketbra{\tn{W}}_N)
\end{align}
is the state distributed to the parties at Step 3 of Protocol~{1}.

\subsection{Security in the presence of noise}

\textbf{Perfect security.} In what follows we will show that our protocol is perfectly secure in the semiactive adversary scenario in the noisy network defined by Eq.~\eqref{eq:n-fold_noise}. We start by defining what it means for a map to preserve permutational invariance.

\begin{definition}[Permutational-invariance preserving map]\label{def:perm_inv_preser_channel}
Let $\pi$ be a permutationally invariant state, such that for all permutations $\Sigma,$
$\pi= \mathcal{V}_\Sigma \de{\pi}$, where $\mathcal{V}_\Sigma $ is a map that performs the permutation $\Sigma$ on the subsystems. A map $\mathcal{E}$ is permutational-invariance preserving if the state after the action of the map $\pi' = \mathcal{E}(\pi)$ is permutationally invariant, i.e., $\pi'= \mathcal{V}_\Sigma \de{\pi'}$.
\end{definition}

Note that the noise channel of our interest, $\Lambda^{\otimes N}$, preserves permutational invariance according to the above definition, due to the tensor structure.

\begin{theorem}\label{thm:WSsecure_Nfold}
The anonymous transmission protocol with the W state, Protocol 1, is sender- and receiver-secure in the semiactive adversary scenario in a noisy network, where noise is defined by Eq.~\eqref{eq:n-fold_noise}. 
\end{theorem}
\begin{proof}	
According to Def.~\ref{def:perm_inv_preser_channel}, the noise channel $\Lambda^{\otimes N}$ is permutational-invariance preserving. Therefore, the proof of Thm.~\ref{thm:WSsecure_Nfold}  follows exactly the same steps as the proof of Thm.~\ref{thm:WSsecure}, where one replaces the state distributed by the source, $\ketbra{\tn{W}}_N$, with $\omega^\Lambda_N$. Therefore if $\rho^\Lambda_{W^\adv C|S=i}$ is the state of the adversaries at the end of the protocol, given that node $i$ is the sender, we have that $\rho^\Lambda_{W^\adv C|S=i} = \rho^\Lambda_{W^\adv C}$, for all $i \notin \adv$, and 
\begin{align}
\begin{split}
& P_{\guess}[S|A,C, S\notin \mathcal{A}] := \\
& \quad = \max_{\{M^i\}} \sum_{i \in [N]}  P[S=i|S\notin \mathcal{A}] \Tr[M^i \cdot \rho^\Lambda_{W^\adv C|S=i} ]\\
& \quad \leq \max_{i \in [N]} P[S=i|S\notin \mathcal{A}].
\end{split}
\end{align}
The same statement holds for receiver-security.
\end{proof}

\textbf{$\varepsilon$ security.} In a realistic quantum network, it is quite unlikely that one will be able to control the noise channels perfectly and ensure that all qubits are subjected to the action of exactly the same noise channel. Here we would like to analyze what happens in the case when the network noise is slightly perturbed, in the sense that each qubit experiences a slightly different noise.
We say that in the perturbed case, the network noise is such that each individual qubit of the multipartite W state, $\ket{\tn{W}}_N$, is subjected to an action of a channel $\Lambda_i$,
\begin{align}\label{eq:n_fold_pert}
\hat{\omega}^\Lambda_N = \bigotimes_{i=1}^N \Lambda_i (\ketbra{\tn{W}}_N),
\end{align}
where $\parallel \Lambda - \Lambda_i \parallel_1 \leq \varepsilon_i$ for some map $\Lambda$, and $\parallel \cdot \parallel_1$ denotes the induced trace norm \cite{Watrous2018}.

Since each channel is slightly perturbed, the state after the action of the channel, $\hat{\omega}^\Lambda_N$, is no longer perfectly permutationally invariant. Yet, intuitively, since the perturbation is small, the state $\hat{\omega}^\Lambda_N$ is
$\varepsilon$-close to a permutationally invariant state, for some small $\varepsilon$, and, consequently, the protocol should be $\varepsilon$-secure. In the following we show that this intuition is, indeed, true. First, let us formalize the notion of $\varepsilon$ security.

\begin{definition}[$\varepsilon$-sender security]
We say that the anonymous transmission protocol is \textit{$ \varepsilon$-sender-secure} if, given that the sender is not the adversary, the probability of the adversaries guessing the sender is
\begin{align}\label{EQ:sec_active}
P_{\guess}[S|W^\adv,C, S\notin \mathcal{A}] \leq \max_{i \in [N]} P[S=i|S\notin \mathcal{A}] + \varepsilon.
\end{align}
And analogously for $\varepsilon$-receiver security.
\end{definition}

\begin{theorem}\label{thrm:n_fold_pert}
The anonymous transmission protocol with the W state, Protocol 1, is $N\varepsilon_{\max}$-sender-secure in the semiactive adversary scenario when the noise in the network is defined by Eq.~\eqref{eq:n_fold_pert}, i.e.,
\begin{align}
\begin{split}
& P_{\guess}[S|W^\adv,C, S\notin \mathcal{A}] \\
& \quad = \max_{\{M^i\}} \sum_{i \in [N]}  P[S=i|S\notin \mathcal{A}] \Tr[M^i \cdot \hat{\rho}^\Lambda_{W^\adv C|S=i} ]  \\ 
& \quad \leq \max_{i \in [N]} P[S=i|S\notin \mathcal{A}] + N\varepsilon_{\max} ,
\end{split}
\end{align}
where $\hat{\rho}_{W^\adv C|S=i}^\Lambda$ is the state of the adversaries at the end of the protocol, and  $\varepsilon_{\max} = \max_{i \in [N]}\varepsilon_i$, with $\varepsilon_i$ given by Eq.~\eqref{eq:n_fold_pert}.
\end{theorem}

The idea of the proof is to show that, for all $i \in [N]$, the trace $\Tr[M^i \cdot \hat{\rho}^\Lambda_{W^\adv C|S=i} ]$ is upper-bounded by $\Tr[M^i \cdot \rho^\Lambda_{W^\adv C|S=i} ] + N\varepsilon_{\max}$. Then using the fact that $N\varepsilon_{\max}$ is independent of $i$, the rest of the proof follows from Thm.~\ref{thm:WSsecure_Nfold}. For details see App. \ref{app:sec:security_analysis}.

\subsection{Performance in a noisy network}

In this section we analyze the performance of Protocol~{1} in a noisy quantum network. To do so reliably, we assume honest implementation; i.e., all of the parties follow the protocol. In the honest implementation, given success in the protocol, the anonymous entangled state between $S$ and $R$ after Step 5. is 
\begin{align}
\omega_{SR} = \frac{1}{\mathcal{N}} \Tr_{N-2}\left[\Lambda^{\otimes N} (\ketbra{\tn{W}}_N) \cdot \left(\mathbb{1}_{SR} \otimes|\vec{0}\rangle\!\langle\vec{0}|_{N-2}\right)\right],
\end{align}
where $\ketbra{\tn{W}}_{N}$ is the $N$-partite W state, $|\vec{0}\rangle\!\langle\vec{0}|_{N-2}$ is a projection onto the $\ket{0}$ state of $N-2$ parties and $\mathcal{N}$ is a normalization factor. Note that in the case where no noise is present we recover the maximally entangled state, i.e. $\omega_{SR} = \ketbra{\psi^+}$, where $\ket{\psi^+} = \frac{1}{\sqrt{2}}(\ket{01} + \ket{01})$.

Throughout the rest of the paper, we will be interested in discussing the performance of anonymous transmission protocols under two types of noise:
\begin{enumerate}
\item $\Lambda$ is the dephasing channel
\begin{align}
\Lambda(\rho) = \mathcal{P}_q(\rho) = q\rho + (1-q) Z\rho Z,
\end{align}
where $\rho$ is a single qubit state, $Z$ is the Pauli $Z$ gate, and $q\in[0,1]$ is the noise parameter.
\item $\Lambda$ is the depolarizing channel
\begin{align}
 \Lambda(\rho) = \mathcal{D}_q(\rho) = q\rho + (1-q) \frac{\mathbb{1}}{2},
\end{align}
where $\rho$ is a single qubit state, $ \frac{\mathbb{1}}{2}$ is a maximally mixed single-qubit state, and $q\in[0,1]$ is the noise parameter. 
\end{enumerate}

\vspace{1em}
\textbf{Comparison with the GHZ protocol \cite{Christandl2005}}. In the following we are interested in comparing the performance of our protocol using the W state with the protocol that uses the GHZ state (for reference see \cite{Christandl2005,Brassard2007}). The main differences between our protocol and the protocol presented in \cite{Christandl2005} lie in \textit{(i)} the initial resource state: W in our case and GHZ for \cite{Christandl2005}; \textit{(ii)} the measurement basis: standard basis for our protocol and $X$ basis for \cite{Christandl2005}; \textit{(iii)} the fact that our protocol is probabilistic, whereas the one with the GHZ state continues regardless of the measurement outcome.

For the noise under consideration, all measurement outcomes in the GHZ protocol are equally likely and the resulting states are equivalent up to a local unitary operation. Therefore, without loss of generality, we consider the state between $S$ and $R$ created in this protocol to be
\begin{align}
\begin{split}
\gamma_{SR} & = \frac{1}{\mathcal{N}'} \Tr_{N-2}\Big[\Lambda^{\otimes N} (\ketbra{\text{GHZ}}_N) \\
& \qquad \qquad \qquad \cdot \left(\mathbb{1}_{SR} \otimes \ketbra{\vec{+}}_{N-2}\right)\Big],
\end{split}
\end{align}
where $\ketbra{\text{GHZ}}_N$ is the $N$-partite GHZ state, $\ketbra{\vec{+}}_{N-2}$ is a projection onto the $\ket{+}$ state of $N-2$ honest parties and $\mathcal{N}'$ is a normalization factor. In the case where no noise is present in the network, the ideal state of $S$ and $R$ is the maximally entangled state $\gamma_{SR} = \ketbra{\phi^+}$, with $\ket{\phi^+} = \frac{1}{\sqrt{2}}(\ket{00} + \ket{11})$. Note that this is a different maximally entangled state than in our W state protocol, but both states are equally useful for teleportation.

To compare the performance of the two protocols, we fix the figure of merit to be the \textit{fidelity} of the obtained anonymous entangled (AE) state with the ideal state that is obtained in the protocol when no noise is present, 
\begin{align}
\label{eq:fid_omega} F_{AE}(\omega_{SR}) &= \Tr[\omega_{SR} \cdot \ketbra{\psi^+}] \\
\label{eq:fid_gamma} F_{AE}(\gamma_{SR}) &= \Tr[\gamma_{SR} \cdot \ketbra{\phi^+}] 
\end{align}
where $\omega_{SR}$ and $\gamma_{SR}$ are anonymous entangled states between $S$ and $R$ arising from measuring W and GHZ states subjected to the network noise.

{In what follows we define what it means for an anoymous entangled state to be useful. Before that, let us motivate it twofold. First, not all states are entangled enough to be a resource for teleportation. It has been shown in \citep{Horodecki1999} that any two-qubit entangled state can be used for teleportation if and only if its singlet fidelity exceeds $\frac{1}{2}$. Secondly, note that the quality of a low-fidelity anonymous entanglement could be further improved by performing entanglement distillation \cite{Deutsch1996} -- a protocol which  creates an entangled state with high fidelity out of a few lower-fidelity states. However, entanglement distillation protocols can be carried out only when fidelities of initial states are larger than $\frac{1}{2}$.  
{We remark that performing entanglement distillation without compromising security of anonymous transfer requires support of anonymous two-way classical communication between $S$ and $R$. This can be achieved, for example, by using a classical anonymous broadcast protocol \cite{Broadbent2007}}.}

{We are now ready to define what it means to say that a resource state is useful for anonymous transmission.}

\begin{definition}[Usefulness]\label{def:FAE_useful}
We say that the anonymous entangled state is a \textit{useful} resource for transmission of a quantum message if its fidelity is strictly larger than $\frac{1}{2}$, i.e. $F_{AE}>\frac{1}{2}$. 
Therefore an $N$-partite state is a useful resource state for anonymous transmission if, upon the parties acting honestly, it can generate anonymous entanglement between any two nodes with $F_{AE}>\frac{1}{2}$.
\end{definition}

To evaluate the behavior of the protocols, we calculate the fidelity of anonymous entanglement as a function of the noise parameter $q$ and the number of nodes $N$, for the depolarizing and dephasing channels. {Examples of the performance of the W and GHZ protocols }for $N=\DE{4,10,50}$ are shown in Fig.~\ref{fig:performance}.

\begin{table}[H]
\centering
\bgroup
\def\arraystretch{2.5}
  \begin{tabular}{ccc} 
  & $F_{AE}(\omega_{SR})$ &$F_{AE}(\gamma_{SR})$ \\ \hline
    Dephasing noise $\mathcal{P}_q^{\otimes N}$ & $1-2q(1-q)$ & $\frac{1+(2q-1)^N}{2}$ \\ 
    Depolarizing noise $\mathcal{D}_q^{\otimes N}$& $\frac{(1+q)(N(q-1)^2 + 4q(1+q))}{4(N(1-q)+4q)}$ & $\frac{2q^N+q^2+1}{4}$ 
  \end{tabular}
  \egroup
\end{table}

\begin{figure} [t]
  \centering
  \includegraphics[width=0.85\linewidth]{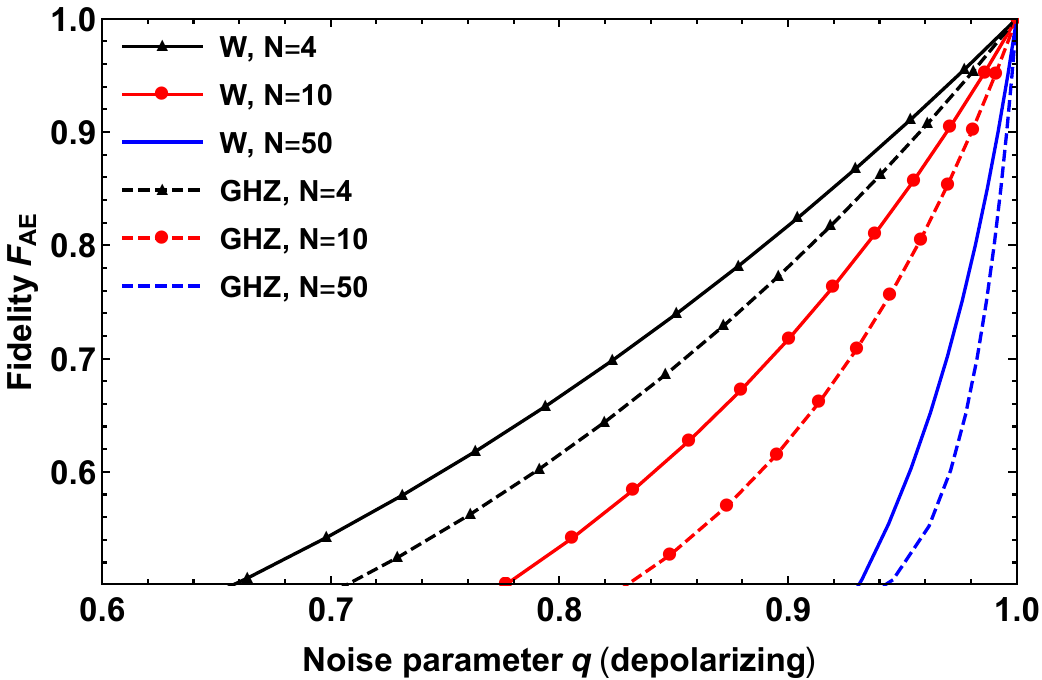}  
\caption{Fidelity of anonymous entanglement as a function of the noise parameter $q$ for depolarizing network noise. Examples for $N=\DE{4,10,50}$.}\label{fig:performance}
\end{figure} 

{We can now ask ourselves which of the states, GHZ or W, tolerates more noise. Note that if one has access to both parameters of the network, noise parameter $q$ and number of nodes $N$, it is easy to determine which of the states would perform better by simply looking at values of $F_{AE}$ calculated from our analytical expressions.}

We start by looking at the dephasing noise. Observe that in this case the fidelity of anonymous entanglement created with the W state $F_{AE}(\omega_{SR})$ is constant in $N$. Specifically, this implies that when fixed dephasing noise is present in the network, the quality of the anonymous link is always the same, regardless of the number of nodes $N$. Moreover, for the dephasing noise, one can observe that $F_{AE}(\omega_{SR})\geq F_{AE}(\gamma_{SR})$ for all $N\geq 2$ and all $q$, which implies that our Protocol~{1} tolerates more noise than the GHZ-based protocol \cite{Christandl2005,Brassard2007}.
  
{When depolarizing noise is present in the network, unlike for the dephasing noise, the fidelity of the anonymous entanglement generated by Protocol~{1} decreases as the number $N$ of parties  increases.
Let us define the noise threshold $q^*$ as the minimum value of noise parameter $q$ for which the anonymous entangled state is still useful in the sense of Def.~\ref{def:FAE_useful}.
One can see that, for small networks (e.g., $N<50$), the threshold $q^*$ is lower for the W state than for the GHZ state $q^*_W < q^*_{GHZ}$, see Fig.~\ref{fig:threshold}, which implies that the $W$ state tolerates more noise in these cases. {However, for $N\geq 182$ one finds that the converse is true, $q^*_W > q^*_{GHZ}$, and therefore the GHZ-based protocol tolerates more noise in this regime.} 
Nevertheless, in App.~\ref{app:sec:realistic_performance} we show that for $N\geq 182$ and larger values of $q$, $q>q^*_W$, we still recover the behavior $F_{AE}(\omega_{SR}) \geq F_{AE}(\gamma_{SR})$. Lastly, we remark 
that the challenge to create a multipartite state scales with the number of parties. Therefore, applications of anonymous transmission of interest in the near future will likely be in the range of 
 $N<50$, in which case Protocol~{1} has proven to be the most noise-tolerant.}

\begin{figure}
  \centering
  \includegraphics[width=0.85\linewidth]{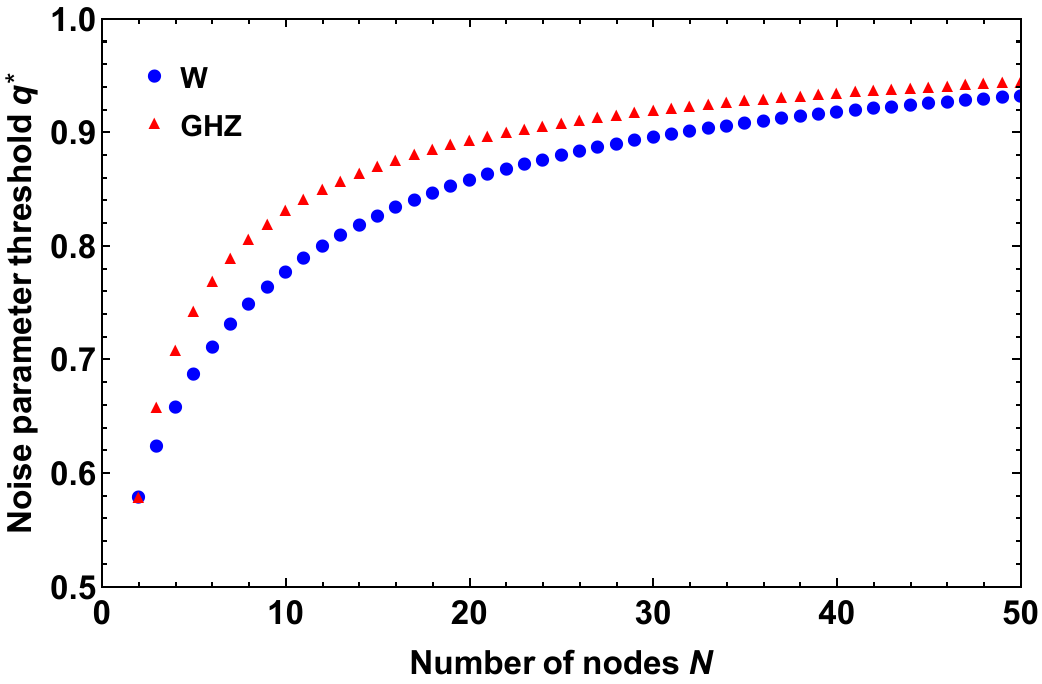}
\caption{Depolarizing parameter thresholds for fidelity of anonymous entanglement $F_{AE} = \frac{1}{2}$.}\label{fig:threshold}   
\end{figure}

Let us also comment on the probability of success of our protocol in the presence of noise. Recall that a round of the protocol only succeeds if in Step 3 the measurement outcome of the  $N-2$ measuring parties is 0. For the dephasing noise the probability of success in our protocol remains $\frac{2}{N}$, which is due to the fact that the noise commutes with the measurement basis. However, for the depolarizing noise the probability of success drops exponentially in $N$. In contrast, for the GHZ state, the outcomes do not need to be post-selected, therefore the protocol \cite{Christandl2005} remains deterministic.

\vspace{1em}
\textbf{Comparison with the relay protocol \cite{Yang2016}.} We now compare our protocol to a scheme proposed in Ref.~\cite{Yang2016}, {which only requires the creation of local Bell pairs and therefore could potentially offer an advantage for a quantum network  implementation}. 
{The main idea of the relay protocol \cite{Yang2016} is
to locally prepare and transmit Bell pairs in order to create a four-partite GHZ state, which will then be turned into anonymous entanglement.}

In the protocol proposed in Ref. \cite{Yang2016}, the nodes are consecutively ordered and each node locally prepares a Bell pair. The first node sends half of her Bell pair to the second node. The second node performs entanglement swapping with a half of her own Bell pair and sends the other half of the state to the next node. This relay continues until the last $N$-th node is reached. $S$ and $R$, however, perform an additional CNOT operation, where they locally entangle the state received from another node with an additional qubit initiated in $\ket{0}$. At the end of this relay
a four-partite GHZ state is created among $S$, $R$, the first and the last node. Finally, anonymous entanglement is established after the first and  the last node perform a measurement.

We explore a scenario for $N = 6$ nodes, assuming that the network is such that quantum channels between parties are depolarizing channels $\Lambda = \mathcal{D}_q$; i.e., whenever a qubit is sent from one party to another it is subject to depolarization. We calculate fidelities of anonymous entanglement for different locations of the $S$ and $R$ in the network. Our results are summarized in App.~\ref{app:sec:realistic_performance}. The numerical evidence shows that in the presence of the depolarizing noise in the network, the fidelity of anonymous entanglement is different depending on the ordering of $S$ and $R$ in the network. 
Note that this does not necessarily imply that the security of the protocol is broken, in the sense that nodes can learn the identity of $S$ and $R$. However, we can see that the performance of the protocol strongly depends on who is sender and receiver, which is not a desirable feature for the anonymous transmission task.

With this in mind, we define the usefulness of the anonymous entanglement created with the relay scheme as the worst case fidelity achieved by the scheme. This is practical if one wants to make sure that the scheme achieves at least a certain fidelity threshold. We then compare the behavior of the relay scheme with the behavior of Protocol {1} in the presence of depolarizing noise. In Fig.~\ref{fig:relay} one can see that in the presence of the depolarizing noise in the network the relay protocol achieves lower fidelity than both the GHZ and the W state protocols.

\begin{figure}[t]
\centering
\includegraphics[width=0.85\linewidth]{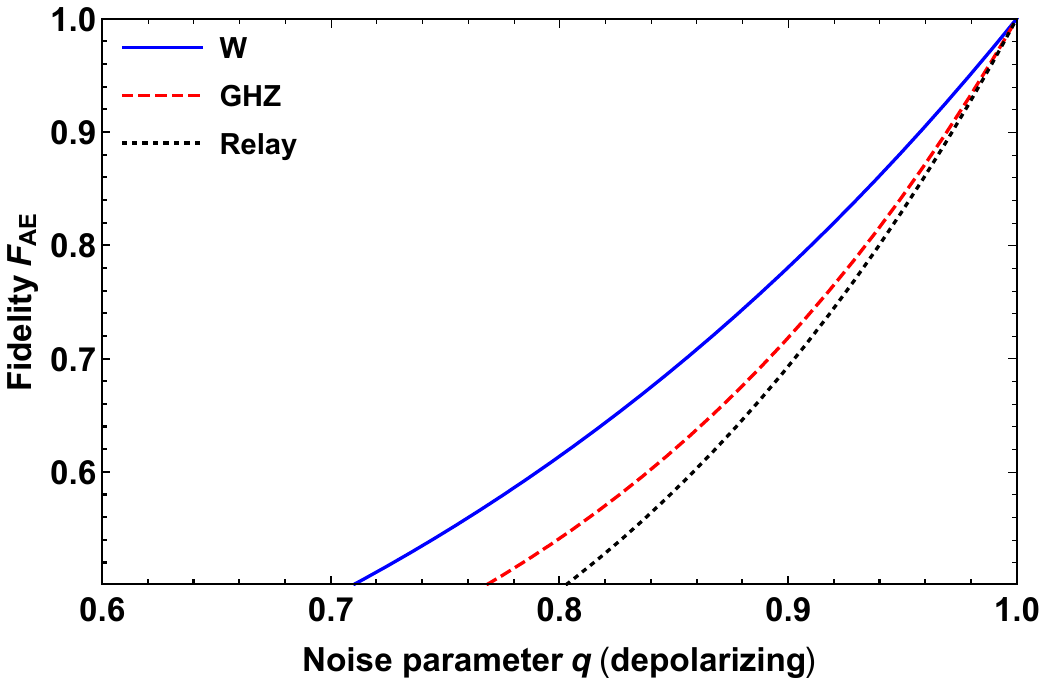}
\caption{Comparison of the fidelity of anonymous entanglement $F_{AE}$ for the W state protocol (Protocol {1}), the GHZ protocol \cite{Christandl2005} and the relay scheme \cite{Yang2016} for $N=6$ nodes.}\label{fig:relay}
\end{figure}

\vspace{1em}
\textbf{Nonresponsive nodes.} {Finally, let us consider the scenario where some of the nodes, that are neither $S$ nor $R$, stop responding. This can happen, for example, due to particle losses in the multipartite state. Note that if $S$ or $R$ lose their particle the teleportation cannot be carried out and, therefore, the protocol is not correct.} 

{Let us consider that the resource state prepared by the source suffers from the action of a noise channel where particles might get lost. Then, with some probability $k$ out of $N$ nodes experience particle loss. Here we ask the question of how many particles losses can be tolerated in an anonymous transmission protocol. Say that a protocol tolerates $k'$ particle losses. After the distribution of the state, if $k$ particles are lost: \textit{(i)} the nodes abort the protocol if $k>k'$, or \textit{(ii)} the remaining $N-k$ parties proceed with the protocol if $k\leq k'$.}

It is known that the entanglement of the GHZ state is not robust to particle losses; \textit{i.e.}, if one particle is lost the remaining $N-1$ parties are left with a separable state.
On the other hand, if the W state is subjected to $N-2$ particle losses the remaining bipartite state is still entangled. In fact, the W state is the most robust to particle losses among all $N$ qubit states \cite{Koashi2000}.
Motivated by this property of the W state, we show that Protocol {1} can tolerate one nonresponsive node. Observe that the $N$-partite W state has the following form after tracing out $k$ out of $N$ parties,
\begin{align}\label{eq:reduced_W}
\Tr_{k}\ketbra{\tn{W}}_N = \frac{N-k}{N}\ketbra{\tn{W}}_{N-k} + \frac{k}{N} |\vec{0}\rangle \! \langle\vec{0}|_{N-k}
\end{align}
where $\ketbra{\tn{W}}_{N-k}$ is the W state of $N-k$ parties.

{In the following theorem we show that Protocol~1 tolerates one particle loss.}

\begin{theorem}
Protocol~1 tolerates one nonresponsive node $i\in[N]\setminus\DE{S,R}$ to produce useful anonymous entanglement, regardless of the number of parties. 
\end{theorem}

\begin{proof} {The proof of the above theorem involves two steps. We first show the correctness of Protocol~{1} when one of the nodes stopped responding, and then show that the created entangled link between $S$ and $R$ is in fact anonymous, i.e. that the security is preserved.}

Let us look at the correctness. The measurement of the state \eqref{eq:reduced_W} in the standard basis and after obtaining all 0 outcomes on $N-k-2$ parties yields a normalized state
\begin{align}\label{eq:Wstate_lossy}
\tilde{\omega}_{SR} = \frac{2}{2+k} \ketbra{\psi^+} + \frac{k}{2+k} \ketbra{00}
\end{align}
which has entanglement fidelity $F_{AE}(\tilde{\omega}_{SR}) = \frac{2}{2+k}$. By Def. \ref{def:FAE_useful} the state $\tilde{\omega}_{SR}$ is useful for anonymous transmission if $\frac{2}{2+k} > \frac{1}{2}$ which implies $k<2$. This yields the desired result. 

{To show that the created entanglement is anonymous, observe that when one of the nodes stops responding the resource state is the state from Eq.~\eqref{eq:reduced_W} with $k=1$. This state is invariant under permutations of nodes and, therefore, we can treat it as a new resource state. Then the security proof follows the same pattern as the proof of Theorem \ref{thm:WSsecure}. }
\end{proof}

For completeness, in App.~\ref{app:sec:realistic_performance} we provide analytical expressions for the fidelity of anonymous entanglement when the W state is subjected to one particle loss, as well as dephasing and depolarizing noise. Fig.~\ref{fig:performance_loss} shows the comparison of anonymous entanglement fidelity of Protocol~{1} under depolarizing noise without particle loss, $F_{AE}({\omega}_{SR})$, and when one particle is lost, $F_{AE}(\tilde{\omega}_{SR})$, for  $N = \DE{4,10,50}$ nodes.
Note that with the growing number of nodes the fidelity of anonymous entanglement in the lossy case approaches the one with no-loss. Indeed, the larger $N$ the smaller the admixture of the $|\vec{0}\rangle \! \langle\vec{0}|_{N-1}$ term in Eq.~\eqref{eq:reduced_W}, and so, with growing $N$ the fidelity is less affected by the loss of a particle. On the other hand, for a larger number of nodes more than one particle loss is more likely to occur. Therefore, the probability that the protocol aborts also increases with the number of nodes.

Lastly, we point out that when one particle is lost in the protocol of Ref. \cite{Yang2016}, the relay cannot be completed. Therefore, much like the GHZ protocol, the relay protocol also cannot be used to create anonymous entanglement whenever one of the nodes is not responsive.

\begin{figure} [t]
\centering
  \centering
  \includegraphics[width=0.85\linewidth]{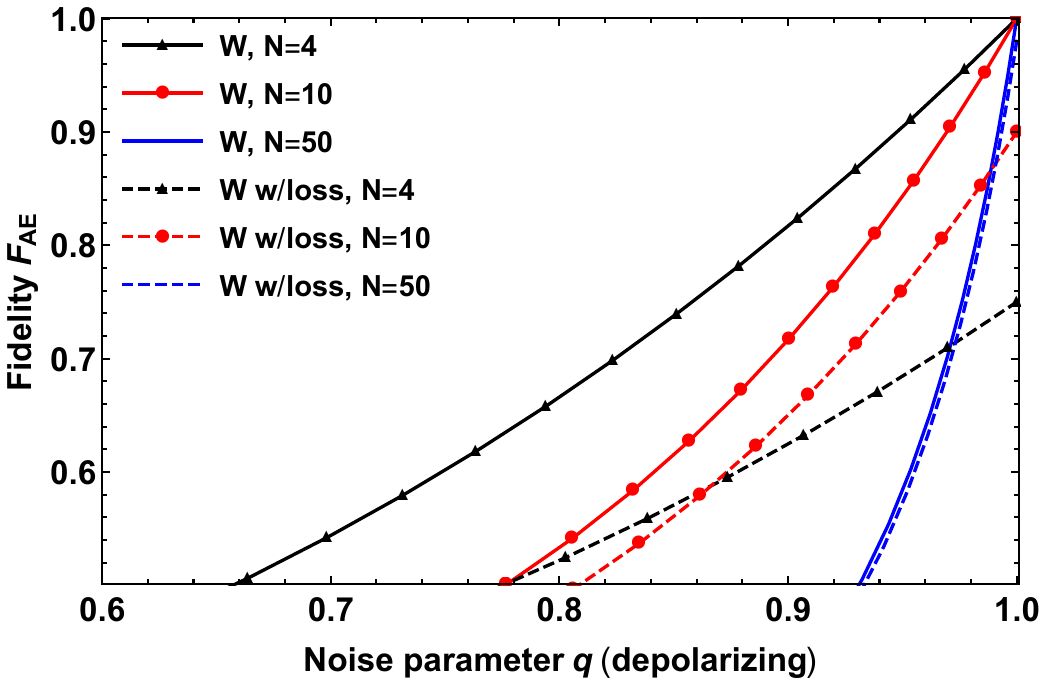}  
\caption{Fidelity of anonymous entanglement as a function of the noise parameter $q$ for depolarizing network noise when the resource W state is subjected to one particle loss. Examples for $N=\DE{4,10,50}$.}\label{fig:performance_loss}
\end{figure} 
\vspace{1em}

\section{Outlook}

We presented a protocol for quantum anonymous transmission using the W state, and proved its security in the semiactive adversary scenario, i.e. when the adversary is active and the source of a quantum state is trusted. 
Moreover, 
we analyzed the behavior of our protocol under the action of common noise models that occur in a realistic quantum network.
\vl{An important question is whether our security proof can be extended to the case where the source might be corrupted, i.e. the fully active adversary scenario. Note that to achieve full security in the noiseless case for the GHZ protocol, Refs. \cite{Brassard2007,Pappa2012} introduced a certification step of the resource state shared by the trusted parties. We remark that for the noiseless W state protocol, it may be possible to achieve full security in a similar way by employing self-testing techniques \cite{Supic2017,Fadel2017}. The problem of certifying the resource state in the presence of noise in the network remains an open question.}

{We have also analyzed the security of our protocol when each qubit suffers the action of a noise channel with slightly different parameters. This bound, however, may not be tight, so another interesting question is whether the security proof can be improved and a stronger bound can be derived  for this case. }

Finally, we have seen that in many instances our W-state based protocol outperforms the GHZ-state and Bell-pair based protocols. For the values of parameters $N$ and $q$, where all the protocols produce useful anonymous entanglement, we remark that a more refined comparison of their performance should take into account the generation rates and resources required to produce the states in every particular experimental setup.

\section{Acknowledgments}

We would like to thank J. Ribeiro, V. Caprara Vivoli, A. Dahlberg, F. Rozp\k{e}dek, I. Kerenidis and E. Diamanti for valuable discussions and insights. We also thank K.  Chakraborty, B. Dirkse, M. Steudtner and K. Goodenough for feedback on the manuscript. This work was supported by STW Netherlands, NWO VIDI, ERC Starting Grant and NWO Zwaartekracht QSC.


%

\newpage
\appendix
\begin{widetext}

\section{Security}\label{sec:security}
\subsection{Classical subroutines}

Our anonymous transmission protocol, Protocol {1}, is built on a few classical subroutines. As mentioned, in Ref. \cite{Broadbent2007}, protocols for implementing these classical subroutines were proposed. Here we list the protocols which we will use as building blocks of our anonymous transmission protocol:

\begin{theorem}[collision detection \cite{Broadbent2007}]
There exists an information-theoretically secure collision detection protocol that takes as input the classical register $Cd_\inn$ of all the participants, $Cd_\inn^i=1$ if node $i$ wishes to be a sender and  $Cd_\inn^i=0$ otherwise, and outputs $Cd_\out=0$ if only one register wants to be the sender and  $Cd_\out=1$ otherwise.
\end{theorem}

\begin{theorem}[receiver notification \cite{Broadbent2007}]
There exists an information-theoretically secure  receiver notification protocol that takes as input the classical register $Rn_\inn$ of the participants and outputs $Rn_\out$, where $Rn_\out^{R}=1$ for the receiver, and all the other parties get output 0.
\end{theorem}

\begin{theorem}[veto \cite{Broadbent2007}]
There exists an information-theoretically secure  veto protocol that takes as input the classical register $O_\inn$ of the parties and outputs $O_{\out}=0$ if all the parties input 0, $O_\inn=\vec{0}$, and  $O_{\out}=1$ otherwise.
\end{theorem}

\begin{theorem}[logical OR \cite{Broadbent2007}]
There exists an information-theoretically secure  logical OR protocol that takes as input the classical register $T_\inn$  and publicly outputs $T_{out}=\oplus_{i=1}^N T_\inn^{i}$.
\end{theorem}

The protocols are information-theoretically secure, in the sense that they do not reveal any classical information other than the one specified by the protocol. The security holds even with an arbitrary number of corrupted participants, assuming the parties share pairwise authenticated private channels and a broadcast channel. 
However, security against a quantum adversary was not analyzed. Here we assume that the protocols listed above remain secure even in the presence of a quantum adversary. This assumption is made explicit in Section \ref{app:sec:states_and_registers} where we assume that the classical subprotocols only act on the classical input register and create the output register, therefore not revealing any information other than what is specified by the protocol, also in the quantum setting.\vspace{1em}

\subsection{States and registers}\label{app:sec:states_and_registers}
In what follows we make
a detailed description of the state in each step of Protocol {1}. Our main goal is to show that the quantum state of the adversary at the end of the protocol does not depend on who is the sender or the receiver. We will later use this fact in the security proof in Sec. \ref{sec:security_analysis}.
  
  Here we adopt the notation that $A$ denotes registers held by the adversary $\adv$, and $\bar{A}$ denotes all the other registers, i.e., of the honest parties (including the sender and the receiver). After Step 2, i.e., once $S$ and $R$ are defined, we distinguish $S$ and $R$ registers from the registers of honest parties $\hon$.

\begin{table}[H]
\centering
\caption{Registers available to parties at each step of the Protocol {1}. All registers are classical unless specified otherwise.}
\bgroup
\def\arraystretch{1.5}
\begin{tabular}{ |c|c|l| } 

 \hline
Step & Available registers & Description\\ \hline
0. & $A_0,\bar{A}_0$ & Quantum side information of dishonest and honest parties before the beginning of Protocol {1}. \\ \hline
1. & $Cd_\inn^{A}$, $Cd_\inn^{\bar{A}}$ & Private input of the parties in the collision detection protocol. \\
& & The node which wants to be a sender inputs 1, the rest 0. \\
& $Cd_\out^{A}$, $Cd_\out^{\bar{A}}$ & Outputs of the collision detection protocol. \\ \hline
2. & $Rn_\inn^{A}$, $Rn_\inn^{\bar{A}}$ & Private input of the receiver notification protocol.\\
& & $S$ inputs the identifier of $R$, everyone else 0. \\
& $Rn_\out^{A}$, $Rn_\out^{\bar{A}}$ & Private outputs of receiver notification protocol. \\ 
& & Output 0 for $R$, 1 for everyone else. \\ \hline
& $D^\adv, D^{\hon S R}$ & Redefined register of dishonest parties $D^\adv=\DE{A_0Cd_{\inn}^{A} Cd_\out^{A}Rn_\inn^{A}Rn_\out^{A}}$\\ & &  and honest parties $D^{\hon S R} =\DE{\bar{A}_0Cd_{\inn}^{\bar{A}} Cd_\out^{\bar{A}}Rn_\inn^{\bar{A}}Rn_\out^{\bar{A}}}$ after Step 2. \\ \hline
3. & $W^\hon,W^\adv,W^S,W^R$ & Quantum registers of the state prepared by the source. \\ \hline
4. & $W^\hon,W^\adv,W^S,W^R$ & Quantum registers of the state prepared by the source. \\ \hline
5. & $O_\inn^\hon$ & Private input of the honest parties to the veto protocol. \\
& & Represented by a string of measurement outcomes $\vec{\nu}$. \\
& $O_\inn^\adv$ & Private input of dishonest parties to the veto protocol. \\
& & Represented by a string of measurement outcomes $\vec{\mu}$. \\
& $O_\out$ & Public output of the veto protocol. \\ 
& & 0 if all entries of strings $\vec{\nu}$ and $\vec{\mu}$ are 0, 1 otherwise. \\ \hline
6. & $Q$ & Quantum register of quantum message $\ket{\psi}$ which $S$ wants to transmit. \\ 
& $T_\inn^S, T_\inn^R$ & Private inputs of $S$ and $R$ to the logical OR protocol. \\
& & $S$ inputs teleportation message $m$ and $R$ inputs random bit rand. \\
& $T_\inn^{\hon},T_\inn^{\adv}$ & Private input of the honest and dishonest parties to the logical OR protocol. \\
& $T$ & Public outcome the logical OR protocol. Outputs XOR of all the inputs. \\
\hline
\end{tabular}\label{app:tab:registers}
\egroup
\end{table}

In the following we specify what are the assumptions associated with each step of the protocol. Additionally, we explicitly write out the state $\xi^{(j)}$ after each step $j$ of the protocol, taking into account all the registers that play a role in the particular step. Therefore, we remark that our notation may be cumbersome at the first glance. However, we advise the reader to refer to Table \ref{app:tab:registers} at any point of our proof.

\subsubsection*{Step 1. Collision detection.}

\begin{assumption}
Let $A_0$ be the quantum side information of dishonest parties and $\bar{A}_0$ be the quantum side information of the honest parties, including sender and receiver, before the beginning of the protocol. We assume that before the start of the protocol the parties share the following state:
\begin{align}
\xi_{A_0 \bar{A}_0 Cd_{\inn} Rn_{\inn}}^{(0)} = \sigma^{(0)}_{A_0\bar{A}_0Cd_{\inn}^{A} Rn_\inn^{A}} \otimes \sigma^{(0)}_{Cd_\inn^{\bar{A}} Rn_\inn^{\bar{A}}}.
\end{align}
In words, we assume the adversaries have a quantum side information, $A_0$, and classical inputs to the collision detection and receiver notification protocol, $Cd_{\inn}^{A}$ and $Rn_\inn^{A}$,  that might be correlated with some quantum side information $\bar{A}_0$ of the remaining parties. However the inputs of the honest parties $Cd_\inn^{\bar{A}}$ and $Rn_\inn^{\bar{A}}$ are uncorrelated with the adversary's state.
\end{assumption}

\begin{assumption}\label{app:assumptionCD}
We assume that the classical collision detection protocol is secure against a quantum adversary, that is, it acts on classical registers $Cd_{\inn}$
and outputs $Cd_\out$ without revealing any other information to the dishonest parties. In particular, if sender and receiver are honest, it does not leak their identity.
\end{assumption}

Let $\xi_{A_0 \bar{A}_0 Cd_{\inn}Cd_\out Rn_{\inn} }^{(1)}$ be the global output state after collision detection (Step 1).
Ass. \ref{app:assumptionCD} implies that tracing out the registers of honest parties (all registers of $\bar{A}$) we obtain a partial state of the adversary (all registers of ${A}$) which is independent of the sender, if the sender is honest. That is, for all honest parties, $\forall i \notin \mathcal{A}$, the state after the collision detection step (Step 1 of Protocol 1) is
\begin{align}
\Tr_{\bar{A}_0 Cd_{\inn}^{\bar{A}} Cd_\out^{\bar{A}} Rn_{\inn}^{\bar{A}}}\de{\xi_{A_0 \bar{A}_0 Cd_{\inn} Cd_\out Rn_{\inn}|S=i}^{(1)}} &= \xi^{(1)}_{A_0Cd_{\inn}^{A} Cd_\out^{A}Rn_\inn^{A}|S=i}\\
&=\xi^{(1)}_{A_0Cd_{\inn}^{A} Cd_\out^{A}Rn_\inn^{A}} .
\end{align}

\subsubsection*{Step 2. Receiver notification.}

\begin{assumption}\label{app:assumptionRN}
We assume that the classical receiver notification protocol is secure against the quantum adversary; that is, the protocol acts on the classical register $Rn_{\inn} $ and outputs $Rn_\out$, without revealing any other information to the dishonest parties. In particular, if sender and receiver are honest, it does not leak their identity.
\end{assumption}

Let the input state to the receiver notification protocol be $\xi_{A_0 \bar{A}_0 Cd_{\inn} Cd_\out Rn_{\inn} }^{(1)}$ and 
the output state conditioned on node $i$ being the sender be
$\xi_{A_0 \bar{A}_0 Cd_{\inn} Cd_\out Rn_{\inn} Rn_\out|S=i}^{(2)}$. 
Ass. \ref{app:assumptionRN} implies that, again, tracing out the registers of honest parties (all registers of $\bar{A}$) we obtain a partial state of the adversary (all registers of ${A}$) which is independent of the sender. That is, for all honest parties $\forall i \notin \mathcal{A}$, the state after the receiver notification step (Step 2 of Protocol 1) is
\begin{align}
\Tr_{\bar{A}_0 Cd_{\inn}^{\bar{A}} Cd_\out^{\bar{A}} Rn_{\inn}^{\bar{A}}Rn_{\inn}^{\bar{A}} Rn_\out^{\bar{A}} }(\xi_{A_0 \bar{A}_0 Cd_{\inn} Cd_\out Rn_{\inn} Rn_\out|S=i}^{(2)}) &= \xi^{(2)}_{A_0Cd_{\inn}^{A} Cd_\out^{A}Rn_\inn^{A} Rn_\out^{A}|S=i}\\
&=\xi^{(2)}_{A_0Cd_{\inn}^{A} Cd_\out^{A}Rn_\inn^{A}Rn_\out^{A}} .
\end{align}

For clarity, we denote the state after the receiver notification (Step 2), given that node $i$ is the sender, by 
\begin{align}
\xi_{A_0 \bar{A}_0 Cd_{\inn} Cd_\out Rn_{\inn} Rn_\out|S=i}^{(2)} &  \equiv \sigma_{D^{\adv}D^{\hon S R}|S=i}
\end{align}
where $D^\adv=\DE{A_0Cd_{\inn}^{A} Cd_\out^{A}Rn_\inn^{A}Rn_\out^{A}}$ denotes all the registers in possession of the adversary at the end of Step 2. And similarly, $D^{\hon S R}$ denotes the registers of the honest parties. Note that, now that sender $S$ and receiver $R$ are defined, we distinguish them from the subset of honest players. 

\begin{lemma}\label{app:lem:sigma_indep_S}
If $S$ and $R$ are honest, the state of the adversary at the end of the receiver notification protocol does not carry any information about their identity. Let $\sigma_{D^{\adv}|S=i}:=\Tr_{D^{\hon S R}}[ \sigma_{D^{\adv}D^{\hon S R}|S=i}]$; by Ass. 2 and 3 it holds that
\begin{align}
\sigma_{D^{\adv}|S=i}=\sigma_{D^{\adv}|S=j} = \sigma_{D^{\adv}} \quad \forall i,j \notin \mathcal{A},
\end{align}
and
\begin{align}
\sigma_{D^{\adv}|R=i}=\sigma_{D^{\adv}|R=j} = \sigma_{D^{\adv}} \quad \forall i,j \notin \mathcal{A}.
\end{align}
\end{lemma}

\subsubsection*{Step 3. State distribution.} 

\begin{assumption}
The $N$-partite state distributed by a trusted source is $\ketbra{\tn{W}}_{W^\hon W^\adv W^S W^R}$. Here $W^\hon$ is the \textit{quantum} register of the honest parties, $W^\adv$ is the \textit{quantum} register of dishonest parties, and $W^S$ and $W^R$ are \textit{quantum} registers of the sender and receiver. 
\end{assumption}

Therefore, the global state after the source distributed the quantum state (Step 3 of Protocol 1)  is
\begin{align}
\xi_{W^\hon W^\adv W^S W^R D^{\adv}D^{\hon S R}|S=i}^{(3)} = \ketbra{\tn{W}}_{W^\hon W^\adv W^S W^R} \otimes \sigma_{D^{\adv}D^{\hon S R}|S=i}.
\end{align}

\subsubsection*{Step 4. Measurement.}

Step 4 describes a measurement on quantum registers $W^\hon W^\adv$ and creates the classical registers $O_\inn^\hon$ and $O_\inn^\adv$. 
The honest parties perform a projection $\Pi^{\vec{\nu}}_{W^\hon}$ on the $\DE{0,1}$ basis and the string of outcomes $\vec{\nu}$ is recorded on register $O_\inn^\hon$.
The adversaries, however, instead of performing the measurement specified by the protocol, can apply an arbitrary map on their registers and produce a classical outcome $\ketbra{\vec{\mu}}_{O_\inn^\adv}$. This action is descried by applying a map $\mathcal{F}^{\vec{\mu}}_{W^\adv D^\adv}$ labeled by $\vec{\mu}$, which acts on registers $W^\adv D^\adv$ and producing a classical outcome $\ketbra{\vec{\mu}}_{O_\inn^\adv}$ in register $O_\inn^\adv$. Note that this outcome can be a strategy upon which dishonest parties agree and, in particular, it does not have to represent the actual action of the map  $\mathcal{F}^{\vec{\mu}}_{W^\adv D^\adv}$.  
Therefore, the state after the parties perform local measurements (Step 4 of Protocol 1)  is described as, 
\begin{align}
\begin{split}
\xi_{W^\hon W^\adv W^S W^R D^{\adv}D^{\hon S R} O_\inn^\hon O_\inn^\adv|S=i}^{(4)} = \sum_{\vec{\mu},\vec{\nu}} \Pi^{\vec{\nu}}_{W^\hon} \otimes \mathcal{F}^{\vec{\mu}}_{W^\adv D^\adv} (\ketbra{\tn{W}}_{W^\hon W^\adv W^S W^R} \otimes \sigma_{D^{\adv}D^{\hon S R}|S=i})\\ \otimes \ketbra{\vec{\nu}}_{O_\inn^\hon} \otimes \ketbra{\vec{\mu}}_{O_\inn^\adv},
\end{split}
\end{align}
where $\Pi^{\vec{\nu}}_{W^\hon}$ corresponds to a projection of register $W^\hon$ onto the state $\ketbra{\vec{\nu}}$ in the standard basis.

\subsubsection*{Step 5. Anonymous announcement of outcomes.}
Each of the parties inputs their measurement outcome into the veto protocol. In particular, $O_\inn^\hon = \ketbra{\vec{\nu}}_{O_\inn^\hon}$ is a private input of the honest parties and $O_\inn^\adv = \ketbra{\vec{\mu}}_{O_\inn^\adv}$ is a private input of the dishonest parties. 

\begin{assumption}\label{app:ass:veto}
We assume that the classical veto protocol is secure against the quantum adversary; i.e., the veto protocol acts on the classical registers $O_\inn^\hon, O_\inn^\adv$, and only outputs $O_\out=0$ if $ O_\inn^\hon = O_\inn^\adv = \ketbra{\vec{0}}$ and 1 otherwise, and does not reveal any other information. 
\end{assumption}

Then, the state after the veto protocol, where the parties announce their outcomes (Step 5 of Protocol 1), is 
\begin{align}
\xi_{W^\hon W^\adv W^S W^R D^{\hon \adv SR} O_\inn^\hon O_\inn^\adv O_\out|S=i}^{(5)} 
=&  \Pi_{W^\hon}^{\vec{0}} \otimes \mathcal{F}_{{W^\adv D^\adv}}^{\vec{0}} (\ketbra{\tn{W}}_{W^\hon W^\adv W^S W^R}  \otimes \sigma_{D^{\adv}D^{\hon S R}|S=i}) \nonumber \\ 
 &\otimes \ketbra{\vec{0}}_{O_\inn^\hon} \otimes \ketbra{\vec{0}}_{O_\inn^\adv} \otimes \ketbra{0}_{O_\out} \nonumber \\
 + & \sum_{\vec{\mu}\neq 0, \vec{\nu} }\Pi_{W^\hon}^{\vec{\nu}} \otimes \mathcal{F}_{{W^\adv D^\adv}}^{\vec{\mu}} (\ketbra{\tn{W}}_{W^\hon W^\adv W^S W^R} \otimes \sigma_{D^{\adv}D^{\hon S R}|S=i})  \nonumber \\ 
 &\otimes \ketbra{\vec{\nu}}_{O_\inn^\hon} \otimes \ketbra{\vec{\mu}}_{O_\inn^\adv} \otimes \ketbra{1}_{O_\out} 
\end{align}

\subsubsection*{Step 6. Teleportation.}

In Step 6., sender and receiver wish to perform the teleportation.  To do so, the sender performs the Bell state measurement and communicates the classical outcome to the receiver, so that she can correct the teleported state. The classical communication is carried out by using the classical protocol logical OR.

\begin{assumption}\label{app:ass:logicalOR}
The classical logical OR protocol acts on classical registers and does not reveal any information other than the logical OR of the  inputs. 
\end{assumption}

Let $Q$ denote the register of the quantum message which sender $S$ wishes to transmit. More formally, this step consists of applying a map, a Bell state measurement, acting on the registers of the sender $W^S$ and $Q$  and producing a classical message in the public register $T$, followed by the receiver applying a unitary operation according to the outcome $m$ of the Bell measurement. We denote the map that describes the teleportation step by $\mathcal{T}_{W^S W^R Q O_\out \rightarrow W^S W^R Q O_\out T_\inn^S T_\inn^R T}$. Its action is conditioned on the outcome of Step 5., i.e., public output of the veto protocol. We define its action on a state $\phi_{W^S W^R}\otimes \ketbra{\psi}_Q$ as follows,
\begin{align}
\begin{split}
\mathcal{T}_{W^S W^R Q | O_\out = 0 \rightarrow W^S W^R Q O_\out T_\inn^S T_\inn^R T} := &\sum_m \mathcal{R}^m_{W^R} \circ \mathcal{B}^m_{W^S Q}(\phi_{W^S W^R}\otimes \ketbra{\psi}_{Q})\\
& \otimes \sum_{\tn{rand}} \frac{1}{4}\ketbra{m}_{T_\inn^S} \otimes \ketbra{\tn{rand}}_{T_\inn^R}\otimes\ketbra{m \oplus \tn{rand}}_T,
\end{split}
\end{align}
\begin{align}
\mathcal{T}_{W^S W^R Q | O_\out = 1 \rightarrow W^S W^R Q O_\out T_\inn^S T_\inn^R T} := \mathbb{1}_{W^S W^R Q}(\phi_{W^S W^R}\otimes \ketbra{\psi}_{Q}) \otimes \ketbra{\perp}_{T_\inn^S}\otimes \ketbra{\perp}_{T_\inn^R}\otimes \ketbra{\perp}_T.
\end{align}
The map $\mathcal{B}^m_{W^S Q}$ represents the Bell state measurement, on registers $W^S Q$, with outcome $m$, and the map $\mathcal{R}^m_{W^R}$ corresponds to the unitary the receiver applies to correct the teleported state.
The action of the map $\mathcal{T}_{W^S W^R Q O_\out \rightarrow W^S W^R Q O_\out T_\inn^S T_\inn^R T}$ describes  that the state $\ketbra{\psi}_{Q}$ is either teleported to register $W^R$ when $O_\out=0$ or the protocol aborts when $O_\out = 1$, which we represent by the state $\ketbra{\perp}_T$ in register $T$.

However, we note that in this step the adversaries could also deviate from the protocol. In general, they could perform an arbitrary map in their registers and input a string $\vec{\kappa}\neq \vec{0}$ to the logical OR protocol. In that case, the teleportation step can be described as
\begin{align}
\begin{split}
&\mathcal{T}_{W^S W^R W^\adv D^\adv Q | O_\out = 0 \rightarrow W^S W^R W^\adv D^\adv Q O_\out T_\inn^S T_\inn^R T_\inn^{\mathcal{A}} T} 
(\xi_{W^\hon W^\adv W^S W^R Q D^{\hon \adv SR} O_\inn^\hon O_\inn^\adv O_\out|S=i}^{(5)} ):= 
\\
&\quad \quad \sum_{m,\vec{\kappa}} \mathcal{R}^{m \oplus_i \kappa_i}_{W^R} \circ \mathcal{G}_{W^\adv D^\adv}^{\vec{\kappa}} \circ \mathcal{B}^m_{W^S Q}(\Pi_{W^\hon}^{\vec{0}} \otimes \mathcal{F}_{{W^\adv D^\adv}}^{\vec{0}} (\ketbra{\tn{W}}_{W^\hon W^\adv W^S W^R} \otimes \ketbra{\psi}_{Q} \otimes \sigma_{D^{\adv}D^{\hon S R}|S=i})) \\
& \quad \quad  \otimes \sum_{\tn{rand}} \frac{1}{4}\ketbra{m}_{T_\inn^S} \otimes \ketbra{\tn{rand}}_{T_\inn^R}\otimes \ketbra{\vec{\kappa}}_{T_\inn^{\mathcal{A}}}\otimes\ketbra{m \oplus \tn{rand}\oplus_i \kappa_i}_T
 \end{split}
\end{align}
where  $\mathcal{G}_{W^\adv D^\adv}^{\vec{\kappa}}$ represents an arbitrary map the adversaries apply to registers $W^\adv D^\adv$, which is followed by the creation of classical register ${T_\inn^{\mathcal{A}}}$. $\mathcal{R}^{m \oplus_i \kappa_i}_R$ expresses the fact that the receiver now applies a unitary labeled by $m \oplus_i \kappa_i$ instead of $m$.

Note that the map $\mathcal{G}_{W^\adv D^\adv}^{\vec{\kappa}}$ 
only acts on the registers of the adversaries and after the teleportation step (Step 6) no other operations are performed by the honest parties. The security of the protocol is defined in terms of the guessing probability, which takes into account an optimization over all maps on the register of the adversary. Therefore, for the security analysis, we can, without loss of generality, neglect the map $\mathcal{G}_{AD^\adv}^{\vec{\kappa}}$ in the final state, since it is taken into account in the definition of the guessing probability.

Finally, the state after the teleportation protocol (Step 6 of Protocol 1) is 
\begin{align}\label{app:eq:xi_6}
\begin{split}
&\xi_{W^\hon W^\adv W^S W^R Q D^{\hon \adv SR} O_\inn^\hon O_\inn^\adv O_\out T_\inn^S T_\inn^R T_\inn^\adv T|S=i}^{(6)} = \\
& \quad \sum_{m,\vec{\kappa}}  \mathcal{R}^{m \oplus_i \kappa_i}_{W^R} \circ \mathcal{B}^m_{W^S Q}(\Pi_{W^\hon}^{\vec{0}} \otimes \mathcal{F}_{{W^\adv D^\adv}}^{\vec{0}} (\ketbra{\tn{W}}_{W^\hon W^\adv W^S W^R} \otimes \ketbra{\psi}_{Q} \otimes \sigma_{D^{\adv}D^{\hon S R}|S=i}))   \\ 
& \quad \quad  \otimes   \ketbra{\vec{0}}_{O_\inn^\hon} \otimes \ketbra{\vec{0}}_{O_\inn^\adv} \otimes \ketbra{0}_{O_\out}\\
& \quad \quad  \otimes   \sum_{\tn{rand}} \frac{1}{4}\ketbra{m}_{T_\inn^S} \otimes \ketbra{\tn{rand}}_{T_\inn^R}\otimes \ketbra{\vec{\kappa}}_{T_\inn^{\mathcal{A}}}\otimes\ketbra{m \oplus \tn{rand}\oplus_i \kappa_i}_T \\
 \quad & +  \sum_{\vec{\mu}\neq 0, \vec{\nu} }\mathbb{1}_{W^S W^R Q}  (\Pi_{W^\hon}^{\vec{\nu}} \otimes \mathcal{F}_{{W^\adv D^\adv}}^{\vec{\mu}} (\ketbra{\tn{W}}_{W^\hon W^\adv W^S W^R} \otimes \ketbra{\psi}_{Q} \otimes \sigma_{D^{\adv}D^{\hon S R}|S=i})) \\ 
& \quad \quad \otimes  \ketbra{\vec{\nu}}_{O_\inn^\hon} \otimes \ketbra{\vec{\mu}}_{O_\inn^\adv} \otimes \ketbra{1}_{O_\out} \otimes \ketbra{\perp}_{T_\inn^S}\otimes \ketbra{\perp}_{T_\inn^R} \otimes \ketbra{\perp}_{T_\inn^\adv} \otimes \ketbra{\perp}_{T}.
 \end{split}
 \end{align}

\vspace{1em}

Observe, however, that the classical registers $D^{\hon SR}, O_\inn^\hon, T_\inn^S T_\inn^R$ are not further acted upon with any map. Moreover, their content is private, as by Lem. \ref{app:lem:sigma_indep_S} 
and Ass. \ref{app:ass:veto} and \ref{app:ass:logicalOR} no information about it is revealed to the adversary. Since we are interested in the information available to the adversary we will trace out these subsystems.

\begin{lemma}\label{app:lem:reduced_state_final}
Let $C = \{D^\adv,O_\inn^\adv,O_\out,T_\inn^\adv,T\}$ represent all the classical and quantum side information accessible to the adversary at the end of  the protocol. The reduced output state of the anonymous transmission protocol with the W state, where we trace out all private information of the honest parties $\hon$, $S$, and $R$, given that node $i$ is the sender, can be described as follows,
\begin{align}\label{app:eq:statefinal}
\rho_{W^\hon W^\adv W^S W^R Q C|S=i} 
=&  \sum_{m,\vec{\kappa}}  \mathcal{R}^{m \oplus_i \kappa_i}_{W^R} \circ \mathcal{B}^m_{W^S Q}(\Pi_{W^\hon}^{\vec{0}} \otimes \mathcal{F}_{{W^\adv D^\adv}}^{\vec{0}} (\ketbra{\tn{W}}_{W^\hon W^\adv W^S W^R} \otimes \ketbra{\psi}_{Q} \otimes \sigma_{D^{\adv}})) \nonumber \\ 
 &\otimes \ketbra{\vec{0}}_{O_\inn^\adv} \otimes \ketbra{0}_{O_\out} \otimes \ketbra{\vec{\kappa}}_{T_\inn^{\mathcal{A}}} \otimes \frac{\mathbb{1}_T}{4} \nonumber \\
 + & \sum_{\vec{\mu}\neq 0, \vec{\nu} }\mathbb{1}_{W^S W^R Q}  (\Pi_{W^\hon}^{\vec{\nu}} \otimes \mathcal{F}_{{W^\adv D^\adv}}^{\vec{\mu}} (\ketbra{\tn{W}}_{W^\hon W^\adv W^S W^R} \otimes \ketbra{\psi}_{Q} \otimes \sigma_{D^{\adv}})) \nonumber \\ 
 &\otimes \ketbra{\vec{\mu}}_{O_\inn^\adv} \otimes \ketbra{1}_{O_\out} \otimes \ketbra{\perp}_{T_\inn^{\mathcal{A}}} \otimes \ketbra{\perp}_{T} 
\end{align}
where we made use of Lem. \ref{app:lem:sigma_indep_S} and the explicitly wrote that the state of register $T$ is maximally mixed.
\end{lemma}

In summary, Lem. \ref{app:lem:reduced_state_final} represents the state at the end of the protocol, given that the adversaries might have acted arbitrarily in Step 4 and under the assumption that, in particular, the classical protocols do not reveal the identities of the sender and the receiver. We will use this state to prove security in the following section.

\subsection{Security analysis}\label{app:sec:security_analysis}

\subsubsection{Semiactive adversary}

In this section we show that Protocol {1} is sender-secure. The key point of the proof is that security follows from permutational invariance of the state. Before proving Thm. \ref{thm:WSsecure}, we first prove the following useful lemma.

\begin{lemma}\label{app:lem:independent_S}
The reduced quantum state of the adversary at the end of the protocol is independent of the sender, i.e., $\forall i \notin \mathcal{A}$, 
\begin{align}
\rho_{W^\adv C|S=i} = \rho_{W^\adv C}.
\end{align}
\end{lemma}
\begin{proof}
Let us first consider the case where the receiver is not an adversary, $R\notin \mathcal{A}$.

By tracing out we have that 
\begin{align}
{\rho}_{W^\adv C|S=i} = \Tr_{W^\hon W^S W^R Q}[{\rho}_{W^\hon W^\adv W^S W^R Q C|S=i} ], 
\end{align}
where $\rho_{W^\hon W^\adv W^S W^R Q C|S=i} $ is the total state at the end of the protocol \eqref{app:eq:statefinal}, Lem. \ref{app:lem:reduced_state_final}, {given that $i$ is the sender}.  
Since $\mathcal{R}^{m \oplus_i \kappa_i}_{W^R}$ and $\sum_m \mathcal{B}^m_{W^S Q}$ are CPTP, they do not change the trace and thus we can write the first part of Eq. \eqref{app:eq:statefinal} as

\begin{align}
\begin{split}
& \Tr_{W^\hon W^S W^R Q}\Big[  \sum_{m,\vec{\kappa}}  \mathcal{R}^{m \oplus_i \kappa_i}_{W^R} \circ \mathcal{B}^m_{W^S Q}(\Pi_{W^\hon}^{\vec{0}} \otimes \mathcal{F}_{{W^\adv D^\adv}}^{\vec{0}} (\ketbra{\tn{W}}_{W^\hon W^\adv W^S W^R} \otimes \ketbra{\psi}_{Q} \otimes \sigma_{D^{\adv}}))  \\ 
& \qquad\qquad\qquad  \otimes \ketbra{\vec{0}}_{O_\inn^\adv} \otimes \ketbra{0}_{O_\out} \otimes \ketbra{\vec{\kappa}}_{T_\inn^{\mathcal{A}}} \otimes \frac{\mathbb{1}_T}{4} \Big]  \\
& = \Tr_{W^\hon W^S W^R Q}\Big[  \Pi_{W^\hon}^{\vec{0}} \otimes \mathcal{F}_{{W^\adv D^\adv}}^{\vec{0}} (\ketbra{\tn{W}}_{W^\hon W^\adv W^S W^R} \otimes \ketbra{\psi}_{Q} \otimes \sigma_{D^{\adv}})  \\ 
& \qquad\qquad\qquad  \otimes \ketbra{\vec{0}}_{O_\inn^\adv} \otimes \ketbra{0}_{O_\out} \otimes\sum_{\vec{\kappa}} \ketbra{\vec{\kappa}}_{T_\inn^{\mathcal{A}}} \otimes \frac{\mathbb{1}_T}{4} \Big] \\
& = \Tr_{W^\hon}\,\Big[ \Pi_{W^\hon}^{\vec{0}} \otimes \mathcal{F}_{{W^\adv D^\adv}}^{\vec{0}} (\tilde{W}_{W^\hon W^\adv} \otimes \sigma_{D^\adv}) \Big] \otimes \ketbra{\vec{0}}_{O_\inn^\adv} \otimes \ketbra{0}_{O_\out} \otimes \sum_{\vec{\kappa}}\ketbra{\vec{\kappa}}_{T_\inn^{\mathcal{A}}} \otimes \frac{\mathbb{1}_T}{4} 
\end{split}
 \end{align}
where $\tilde{W}_{W^\hon W^\adv}$ is the reduced W state on registers $W^\hon$ and $W^\adv$ after tracing out $W^S$ and $W^R$, i.e. $\tilde{W}_{W^\hon W^\adv} = \Tr_{W^S W^R}(\ketbra{\tn{W}}_{W^\hon W^\adv W^S W^R})$, and similarly for the second term of \eqref{app:eq:statefinal}. So,

\begin{align}
\rho_{W^\adv C|S=i}&
=\Tr_{W^\hon}\,\Big[ \Pi_{W^\hon}^{\vec{0}} \otimes \mathcal{F}_{{W^\adv D^\adv}}^{\vec{0}} (\tilde{W}_{W^\hon W^\adv} \otimes \sigma_{D^\adv}) \Big] \otimes \ketbra{\vec{0}}_{O_\inn^\adv} \otimes \ketbra{0}_{O_\out} \otimes \sum_{\vec{\kappa}} \ketbra{\vec{\kappa}}_{T_\inn^{\mathcal{A}}} \otimes \frac{\mathbb{1}_T}{4} \nonumber \\
&+ \sum_{\vec{\mu}\neq0,\vec{\nu} }\Tr_{W^\hon}\,\Big[ \Pi_{W^\hon}^{\vec{\nu}} \otimes \mathcal{F}_{{W^\adv D^\adv}}^{\vec{\mu}} (\tilde{W}_{W^\hon W^\adv} \otimes \sigma_{D^\adv}) \Big] \otimes  \ketbra{\vec{\mu}}_{O_\inn^\adv} \otimes \ketbra{1}_{O_\out} \otimes \ketbra{\perp}_{T_\inn^\adv} \otimes \ketbra{\perp}_{T} 
\end{align}
But since the state distributed by the source is permutationally invariant, it holds that 
\begin{align}
\tilde{W}_{W^\hon W^\adv}  = \Tr_{W^{S=i} W^R}(\ketbra{\tn{W}}_{W^\hon W^\adv W^{S=i} W^R}) = \Tr_{W^{S=j} W^R}(\ketbra{\tn{W}}_{W^\hon W^\adv W^{S=j} W^R}), \quad \forall i,j \notin \adv
\end{align} 
Since no other part of the state $\rho_{W^\adv C|S=i}$ depends on the sender, the state $\rho_{W^\adv C|S=i}$ must be the same for all senders and we denote $\rho_{W^\adv C|S=i} = \rho_{W^\adv C}$. Note that the same statement holds when the receiver is honest since,
\begin{align}
\Tr_{W^{S} W^{R=i}}(\ketbra{\tn{W}}_{W^\hon W^\adv W^{S} W^{R=i}}) = \Tr_{W^{S} W^{R=j}}(\ketbra{\tn{W}}_{W^\hon W^\adv W^{S} W^{R=j}}), \quad \forall i,j \notin \adv
\end{align}
and, therefore, $\rho_{W^\adv C|R = i} = \rho_{W^\adv C}$.

\vspace{1em}

Now we proceed to the proof of this statement in the case where the receiver is an adversary.

If the receiver is dishonest then the teleportation map has to take into account the fact that the adversaries can apply an arbitrary map instead of $\mathcal{R}^{m \oplus_i \kappa_i}_{W^R}$. Also, now the output of the teleportation $m$ is known to the adversaries and the map $ \mathcal{F}_{{W^\adv D^\adv}}^{\vec{\mu}}$ could initially also act on the receiver's register. Now we can model the action of the receiver after receiving $m$ by an arbitrary map that acts on all the registers in possession of the adversaries, i.e., $\mathcal{R}^{m \oplus_i \kappa_i}_{W^R} \longrightarrow \mathcal{R'}_{W^\adv W^RCT_\inn^\adv T}$ and instead of \eqref{app:eq:statefinal}, the final state of the protocol is described by 
\begin{align}\label{eq:statefinalRcheat}
\rho_{W^\hon W^\adv W^S W^R Q C|S=i} & =  \mathcal{R'}_{W^\adv W^RCT_\inn^\adv T}\circ \Big(\sum_{m,\vec{\kappa}}  \mathcal{B}^m_{W^S Q}(\Pi_{W^\hon}^{\vec{0}} \otimes \mathcal{F}_{{W^\adv D^\adv}}^{\vec{0}} (\ketbra{\tn{W}}_{W^\hon W^\adv W^S W^R} \otimes \ketbra{\psi}_{Q} \otimes \sigma_{D^{\adv}})) \nonumber \\ 
& \qquad\qquad\qquad  \otimes \ketbra{\vec{0}}_{O_\inn^\adv} \otimes \ketbra{0}_{O_\out} \otimes \ketbra{\vec{\kappa}}_{T_\inn^{\mathcal{A}}} \otimes \ketbra{m}_T \Big) \nonumber \\
& + \sum_{\vec{\mu}\neq0,\vec{\nu}}\mathbb{1}_{W^S W^R Q} \circ (\Pi_{W^\hon}^{\vec{\nu}} \otimes \mathcal{F}_{{W^\adv D^\adv}}^{\vec{\mu}} (\ketbra{\tn{W}}_{W^\hon W^\adv W^S W^R} \otimes \ketbra{\psi}_{Q} \otimes \sigma_{D^{\adv}})) \nonumber \\ 
& \qquad\qquad\qquad \otimes  \ketbra{\vec{\mu}}_{O_\inn^\adv} \otimes \ketbra{1}_{O_\out} \otimes \ketbra{\perp}_{T_\inn^{\mathcal{A}}}\otimes  \ketbra{\perp}_T 
\end{align}
Let us look at the reduced final state of the adversary, which now includes the receiver,
$\rho_{W^\adv W^R C|S=i} = \Tr_{W^\hon W^S Q}[\rho_{W^\hon W^\adv W^S W^R Q C|S=i}]$. 
By the permutational invariance of the state generated by the source we have that the state at the end of the protocol given that node $i$ is the sender is equivalent to the state given that node $j$ is the sender up to a permutation of $i$ and $j$,
\begin{align}
\rho_{W^\hon W^\adv W^S W^R Q C|S=i} = \mathcal{P}_{i\leftrightarrow j}(\rho_{W^\hon W^\adv W^S W^R Q C|S=j}).
\end{align}
Therefore tracing out the sender and the other honest parties, the remaining states are equal
\begin{align}
\rho_{W^\adv W^R C|S=i}=\rho_{W^\adv W^R C|S=j},
\end{align}
which proves anonymity of the sender even if the receiver is dishonest.

\end{proof}

\begin{proof}[Proof Thm. \ref{thm:WSsecure} (sender security)]
Here we focus on proving sender security. The receiver security is formally stated in Thm. \ref{app:thm:WSsecureR}G. iven Lem. \ref{app:lem:independent_S}, we have that
\begin{align}
P_{\guess}[S|W^\adv,C, S\notin \mathcal{A}] & = \max_{\{M^i\}} \sum_{i \in [N]}  P[S=i|S\notin \mathcal{A}] \Tr[M^i \cdot \rho_{W^\adv C|S=i} ] \\
& = \max_{\{M^i\}} \sum_{i \in [N]}  P[S=i|S\notin \mathcal{A}] \Tr[M^i \cdot \rho_{W^\adv C} ] \\
& \leq \max_{i} P[S=i|S\notin \mathcal{A}] \max_{\{M^i\}}  \Tr\Big[\underbrace{\sum_{i \in [N]}  M^i}_{{\mathbb{1}_{W^\adv C}}} \cdot \rho_{W^\adv C}\Big] \\
& = \max_{i} P[S=i|S\notin \mathcal{A}] 
\end{align}

\end{proof}

Analogously, we will prove the following statement for the receiver-security.

\begin{theorem}[receiver security]\label{app:thm:WSsecureR}
The anonymous transmission protocol, Protocol 1, with the W state is receiver-secure in the semiactive adversary scenario, i.e. 
\begin{align}
\max_{\{M^i\}} \sum_{i \in [N]}  P[R=i|W^\adv,C, R\notin \mathcal{A}] \Tr[M^i \cdot \rho_{W^\adv C|R=i} ] \leq \max_{i} P[R=i|R\notin \mathcal{A}],
\end{align}
given that the receiver is honest. 
\end{theorem}

\begin{proof}
By the proof of Lem. \ref{app:lem:independent_S}, it follows that the reduced quantum state of the adversary at the end of the protocol is independent of the receiver, i.e.,  $\rho_{W^\adv C|R=i} = \rho_{W^\adv C}, \forall i \notin \mathcal{A}$. Therefore,
\begin{align}
P_{\guess}[R|W^\adv,C, R\notin \mathcal{A}] & = \max_{\{M^i\}} \sum_{i \in [N]}  P[R=i|R\notin \mathcal{A}] \Tr[M^i \cdot \rho_{W^\adv C|R=i} ] \\
& \leq \max_{i} P[R=i|R\notin \mathcal{A}] \max_{\{M^i\}}  \Tr\Big[\underbrace{\sum_{i \in [N]}  M^i}_{{\mathbb{1}_{W^\adv C}}} \cdot \rho_{W^\adv C}\Big] \\
& = \max_{i} P[R=i|R\notin \mathcal{A}]
\end{align}
\end{proof}

\subsubsection{Passive adversary}

\begin{definition}
Let $\mathcal{H}$ be the subset of honest players, excluding $S$ and $R$, and $\mathcal{A}$ be the subset of passive adversaries. Let $C$ be the register that contains all classical information accessible to the adversaries, i.e., the public outputs of the classical subprotocols, plus all the inputs and outputs of the adversaries to these classical subprotocols, $C = \{D^\adv,O_\inn^\adv,O_\out,T_\inn^\adv,T\}$. Then probability of the adversaries guessing the sender is given by
\begin{align}\label{app:EQ:sec_passive}
P_{\guess}[S|W^\adv,C, S\notin \mathcal{A}] = \sum_{a,c}P[W^\adv=a,C=c] \max_{i \in [N]} P[S=i|W^\adv=a,C=c,S\notin \mathcal{A}], \quad 
\end{align}
where maximization is taken over all the values of random variable $S$, and $a$ and $c$ are possible values of random variables $W^\adv$ and $C$ respectively. Note that, unlike before, here $W^\adv$ is a classical register of the adversary, since their share of the $W$ state was measured in the $\DE{0,1}$ basis. An analogous expression holds for receiver-security. 
\end{definition}

The proof for the passive adversary security scenario is a special case of the proof for the semiactive adversary scenario. Indeed, it corresponds to the case where the arbitrary map of the adversary, $\mathcal{F}_{W^\adv \mathcal{D}^{\mathcal{A}}}^{\vec{\mu}}$, is a measurement in the $\DE{\ket{0},\ket{1}}$ basis and $T_\inn^\adv=\vec{0}$. Let us first prove the following lemma.
\begin{lemma} \label{app:lem:independent_S_passive}
The probability of registers $W^\adv$ and $C$ assuming certain values $a$ and $c$ is independent of the sender,
\begin{align}
P[W^\adv=a,C=c|S = i,S\notin \mathcal{A}] = P[W^\adv=a,C=c]
\end{align}
\end{lemma}
\begin{proof}
In the passive adversary scenario, the 
dishonest parties follow the protocol, therefore the map $\mathcal{F}_{W^\adv D^\adv}^{\vec{0}}$ is replaced by a projector onto the $\ketbra{\vec{0}}_{W^\adv}$ subspace, i.e. $\Pi_{W^\adv}^{\vec{0}}$. By the permutational invariance argument the state, in this case classical, is independent of the sender $S$ (or the receiver $R$), which completes the proof.
\end{proof}

\begin{proof}[Proof of Thm. \ref{thm:passive_security}]
Let us expand the probability appearing in the security definition \eqref{app:EQ:sec_passive} 
\begin{align}
P[S=i|W^\adv=a,C=c,S\notin \mathcal{A}] 
& = \frac{P[W^\adv=a,C=c|S=i,S\notin\mathcal{A}] P[S=i|S\notin\mathcal{A}]}{P[W^\adv=a,C=c]}  \\
& = \frac{P[W^\adv=a,C=c|S=i] P[S=i|S\notin\mathcal{A}]}{P[W^\adv=a,C=c]} \label{app:eq:sec_proof_passive}\\
& = P[S=i|S\notin\mathcal{A}]
\end{align}
where in \eqref{app:eq:sec_proof_passive} we used Lem. \ref{app:lem:independent_S_passive}. Therefore, \eqref{app:EQ:sec_passive} becomes,
\begin{align}
P_{\guess}[S|W^\adv,C, S\notin \mathcal{A}] 
& = \sum_{a,c}P[W^\adv=a,C=c] \max_{i \in [N]} P[S=i|S\notin \mathcal{A}] \\
& = \max_{i \in [N]} P[S=i|S\notin \mathcal{A}].
\end{align}
\end{proof}

\section{Anoymous transmission in a noisy quantum network}\label{app:sec:performance}

\subsection{Proof for $\varepsilon$-security}
Here we provide a proof of Thm. \ref{thrm:n_fold_pert} for $\varepsilon$-sender security.

\begin{proof}[Proof of Thm. \ref{thrm:n_fold_pert}]
The idea of our proof is to show that, for all $i$, the trace $\Tr[M^i \cdot \hat{\rho}^\Lambda_{W^\adv C|S=i} ]$ can be upper-bounded by $\Tr[M^i \cdot \rho^\Lambda_{W^\adv C|S=i} ] + N\varepsilon_{\max}$. Then using the fact that $N\varepsilon_{\max}$ is independent of $i$, the rest of the proof follows from Thm. \ref{thm:WSsecure_Nfold}. 

Let us look at the following expression, $\forall i$,
\begin{align}
&\left|\Tr[M^i \hat{\rho}^\Lambda_{W^\adv C|S=i}] - \Tr\Big[M^i \rho^\Lambda_{W^\adv C|S=i}\Big] \right| \nonumber \\
& \leq \norm{ \hat{\rho}^\Lambda_{W^\adv C|S=i} - \rho^\Lambda_{W^\adv C|S=i} }_1 \\
& \leq \norm{\xi'^{\Lambda ~ (6)}_{W^\hon W^\adv W^S W^R Q D^{\hon \adv SR} O_\inn^\hon O_\inn^\adv O_\out T_\inn^S T_\inn^R T_\inn^\adv T|S=i}
- \xi^{\Lambda ~ (6)}_{W^\hon W^\adv W^S W^R Q D^{\hon \adv SR} O_\inn^\hon O_\inn^\adv O_\out T_\inn^S T_\inn^R T_\inn^\adv T|S=i}}_1, \nonumber
\end{align}
where $\xi'^{\Lambda ~ (6)}_{W^\hon W^\adv W^S W^R Q D^{\hon \adv SR} O_\inn^\hon O_\inn^\adv O_\out T_\inn^S T_\inn^R T_\inn^\adv T|S=i}$ and $\xi^{\Lambda ~ (6)}_{W^\hon W^\adv W^S W^R Q D^{\hon \adv SR} O_\inn^\hon O_\inn^\adv O_\out T_\inn^S T_\inn^R T_\inn^\adv T|S=i}$ are final states of the protocol after Step 6 (defined analogously to equation \eqref{app:eq:xi_6}) when the network is perturbed \eqref{eq:n_fold_pert}, or not \eqref{eq:n-fold_noise}, respectively.
Since the protocol is described by a CPTP map, the trace distance of the final state is upper-bounded by the trace distance of the initial state, 
\begin{align}
&\left|\Tr[M^i \hat{\rho}^\Lambda_{W^\adv C|S=i}] - \Tr\Big[M^i \rho^\Lambda_{W^\adv C|S=i}\Big] \right| \nonumber \\
& \leq \norm{ \omega'^\Lambda_{W^\hon W^\adv W^S W^R} \otimes \ketbra{\psi}_{Q}  \otimes \sigma_{D^{\hon \adv SR}|S=i} - \omega^\Lambda_{W^\hon W^\adv W^S W^R} \otimes \ketbra{\psi}_{Q}  \otimes \sigma_{D^{\hon \adv SR}|S=i} }_1\\
&\leq \norm{ \omega'^\Lambda_{W^\hon W^\adv W^S W^R}  - \omega^\Lambda_{W^\hon W^\adv W^S W^R} }_1 \\
&\leq \norm{ \bigotimes_{i=1}^N \Lambda_i (\ketbra{\tn{W}}_{W^\hon W^\adv W^S W^R})  - \Lambda^{\otimes N} (\ketbra{\tn{W}}_{W^\hon W^\adv W^S W^R})}_1 \\
&\leq \norm{ \bigotimes_{i=1}^N \Lambda_i   - \Lambda^{\otimes N} }_1 \leq \sum_{i=1}^N \norm{\Lambda_i - \Lambda }_1  = \sum_{i=1}^N \varepsilon_i \leq N\varepsilon_{\max}
\end{align}
where we used the properties of the trace distance and the induced trace norm. Therefore we have that, $\forall i$
\begin{align}
\Tr[M^i \cdot \hat{\rho}^\Lambda_{W^\adv C|S=i} ] \leq \Tr[M^i \cdot \rho^\Lambda_{W^\adv C|S=i} ] + N\varepsilon_{\max}
\end{align}
so using Thm. \ref{thm:WSsecure_Nfold},
\begin{align}
P_{\guess}[S|W^\adv,C, S\notin \mathcal{A}] 
& = \max_{\{M^i\}} \sum_{i \in [N]}  P[S=i|S\notin \mathcal{A}] \Tr[M^i \cdot \hat{\rho}^\Lambda_{W^\adv C|S=i} ] \\
& \leq \max_{\{M^i\}} \sum_{i \in [N]}  P[S=i|S\notin \mathcal{A}] \left(\Tr[M^i \cdot \rho^\Lambda_{W^\adv C|S=i} ] + N\varepsilon_{\max}\right)\\
& = \max_{\{M^i\}} \sum_{i \in [N]}  P[S=i|S\notin \mathcal{A}] \Tr[M^i \cdot \rho^\Lambda_{W^\adv C|S=i} ] + \sum_{i \in [N]}  P[S=i|S\notin \mathcal{A}] N\varepsilon_{\max} \\
& \leq \max_{i \in [N]} P[S=i|S\notin \mathcal{A}] + N\varepsilon_{\max}.
\end{align}
\end{proof}
The same argument holds for receiver-security. 

\subsection{Performance in a noisy network}\label{app:sec:realistic_performance}

\textbf{Fidelity derivation.} In general, it is non-trivial to derive analytical expressions for fidelity of anonymous entanglement in the presence of noise. The most troublesome part is to obtain analytical expressions for anonymous entangled states shared between $S$ and $R$, which are affected by the noise. Nevertheless, to obtain these explicit formulas, we used the fact that the noise is described by a linear map which acts on each qubit individually. We will illustrate the gist of our derivation with an example for the GHZ state, since it is easier to follow than the one for the W state. 

As defined in the main text, the state shared by $S$ and $R$ in the noisy case is
\begin{align}
\gamma_{SR} = \frac{1}{\mathcal{N}'} \Tr_{N-2}\left[\Lambda^{\otimes N} (\ketbra{\tn{GHZ}}_{N}) \cdot \ketbra{\vec{+}}_{N-2}\right],
\end{align}
where $\mathcal{N}$ is the normalization factor. Note that the GHZ state can be written as
\begin{align}
\ketbra{\tn{GHZ}}_{N} = \frac{1}{2}\left(\ketbra{0}^{\otimes N} + \ketbra{0}{1}^{\otimes N} + \ketbra{1}{0}^{\otimes N} + \ketbra{1}^{\otimes N}\right)
\end{align}
Due to the tensor structure and linearity of the noise, we can write that
\begin{align}
\begin{split}
\gamma_{SR} & = \frac{1}{2\mathcal{N}'} \Tr_{N-2}\left[\left(\Lambda(\ketbra{0})^{\otimes N} + \Lambda(\ketbra{0}{1})^{\otimes N} + \Lambda(\ketbra{1}{0})^{\otimes N} + \Lambda(\ketbra{1})^{\otimes N}\right) \cdot \ketbra{+}^{\otimes N-2}\right] \\
& =\frac{1}{2\mathcal{N}'}\Big( \Tr[\Lambda(\ketbra{0})]^{N-2}\Lambda(\ketbra{0})^{\otimes 2} + \Tr[\Lambda(\ketbra{0}{1})]^{N-2}\Lambda(\ketbra{0}{1})^{\otimes 2} \\
& \qquad \quad + \Tr[\Lambda(\ketbra{1}{0})]^{N-2}\Lambda(\ketbra{1}{0})^{\otimes 2} + \Tr[\Lambda(\ketbra{1})]^{N-2}\Lambda(\ketbra{1})^{\otimes 2} \Big).
\end{split}
\end{align}
This way one only takes the tensor product of the two terms corresponding to $S$ and $R$, instead of taking the tensor of $N$ terms. The expression for the W state follows the exact same pattern, but one has to account for all the combinations of 0's and 1's occurring in the state $\ketbra{\tn{W}}_{N}$. Let $\tn{tr}_{xy} := \Tr[\Lambda(\ketbra{x}{y})\cdot \ketbra{0}]$ with $x,y = \DE{0,1}$. Then the state $\omega_{SR}$ shared between $S$ and $R$ in the noisy implementation of Protocol 1 is
\begin{align}
\begin{split}
\omega_{SR} & = \frac{1}{\mathcal{N}}\Big( (N-2)(N-3)\tn{tr}_{01}\tn{tr}_{10}\tn{tr}_{00}^{N-4} \Lambda(\ketbra{0}) \otimes \Lambda(\ketbra{0})\\
&\qquad  + (N-2)\tn{tr}_{10}\tn{tr}_{00}^{N-3} \big( \Lambda(\ketbra{0}{1}) \otimes \Lambda(\ketbra{0}) + \Lambda(\ketbra{0}) \otimes \Lambda(\ketbra{0}{1}) \big) \\
&\qquad  + (N-2)\tn{tr}_{01}\tn{tr}_{00}^{N-3} \big( \Lambda(\ketbra{1}{0}) \otimes \Lambda(\ketbra{0}) + \Lambda(\ketbra{0}) \otimes \Lambda(\ketbra{1}{0}) \big)\\
&\qquad  + (N-2)\tn{tr}_{11}\tn{tr}_{00}^{N-3} \Lambda(\ketbra{0}) \otimes \Lambda(\ketbra{0}) \\
&\qquad  + \tn{tr}_{00}^{N-2} \big( \Lambda(\ketbra{0}{1}) \otimes \Lambda(\ketbra{1}{0}) + \Lambda(\ketbra{1}{0}) \otimes \Lambda(\ketbra{0}{1}) + \Lambda(\ketbra{0}{0}) \otimes \Lambda(\ketbra{1}{1}) + \Lambda(\ketbra{1}{1}) \otimes \Lambda(\ketbra{0}{0}) \big)\Big).\\
\end{split}
\end{align}
Using the explicit form of $\Lambda$ for the depolarizing and dephasing noise, after easy but tedious calculations, one obtains explicit fidelity expressions derived from Eq. \eqref{eq:fid_omega} and \eqref{eq:fid_gamma}.

\vspace{1em}
\textbf{Dephasing and depolarizing noise.} In this section we provide additional details to the noise analysis provided in the main text. First, we plot the behavior of our protocol vs. the GHZ-based protocol under the dephasing noise, for examples $N = \DE{4,10,50}$, Fig. \ref{app:fig:fid_deph}. Note that the GHZ state is increasingly useful according to Def. \ref{def:FAE_useful} for $q<0.5$. For anonymous entanglement created with the W state this is always the case, however, for the GHZ only for even $N$. To observe the same behavior for odd $N$ and the GHZ state one would have to redefine Eq. \eqref{eq:fid_gamma} to compare the fidelity with the state $\ketbra{\phi^-}$.

\begin{figure} [!h]
\centering
\begin{minipage}[t]{.45\textwidth}
  \centering
  \includegraphics[width=\linewidth]{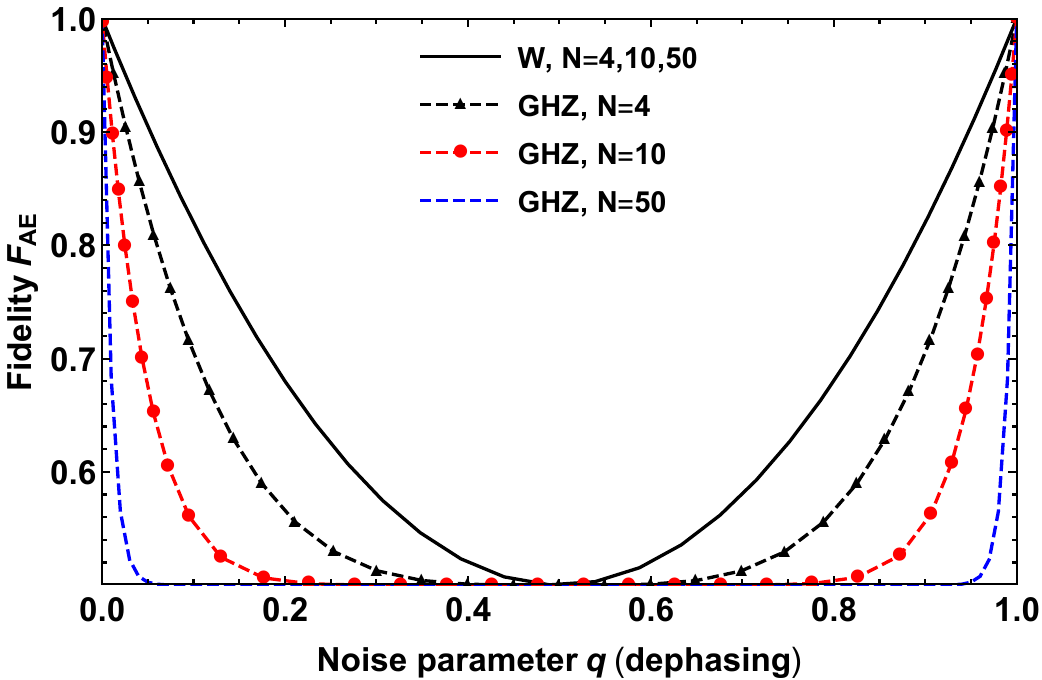}
\caption{Fidelity of anonymous entanglement as a function of the noise parameter for the dephasing channel.}\label{app:fig:fid_deph}
\end{minipage}%
\hspace{0.01\textwidth}
\begin{minipage}[t]{.45\textwidth}
  \centering
  \includegraphics[width=\linewidth]{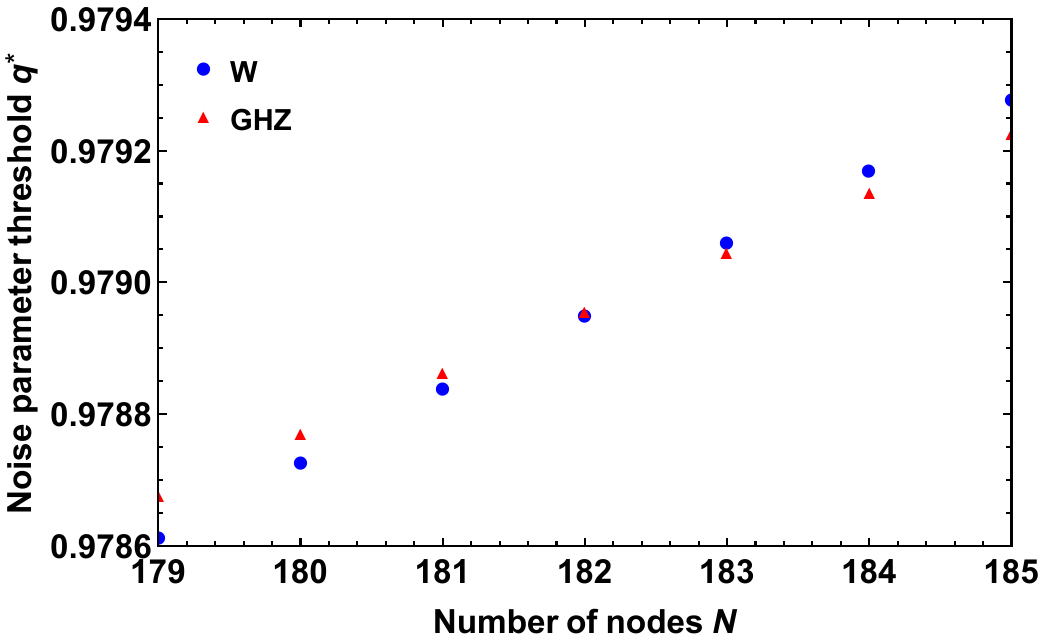}
\caption{Noise parameter threshold for the depolarizing noise. Close-up to $179 \leq N\leq 185$.}\label{app:fig:thresh_close}
\end{minipage}%
\hspace{0.01\textwidth}
\end{figure}

\begin{figure} [!h]
\begin{minipage}[t]{.45\textwidth} 
  \centering
  \includegraphics[width=\linewidth]{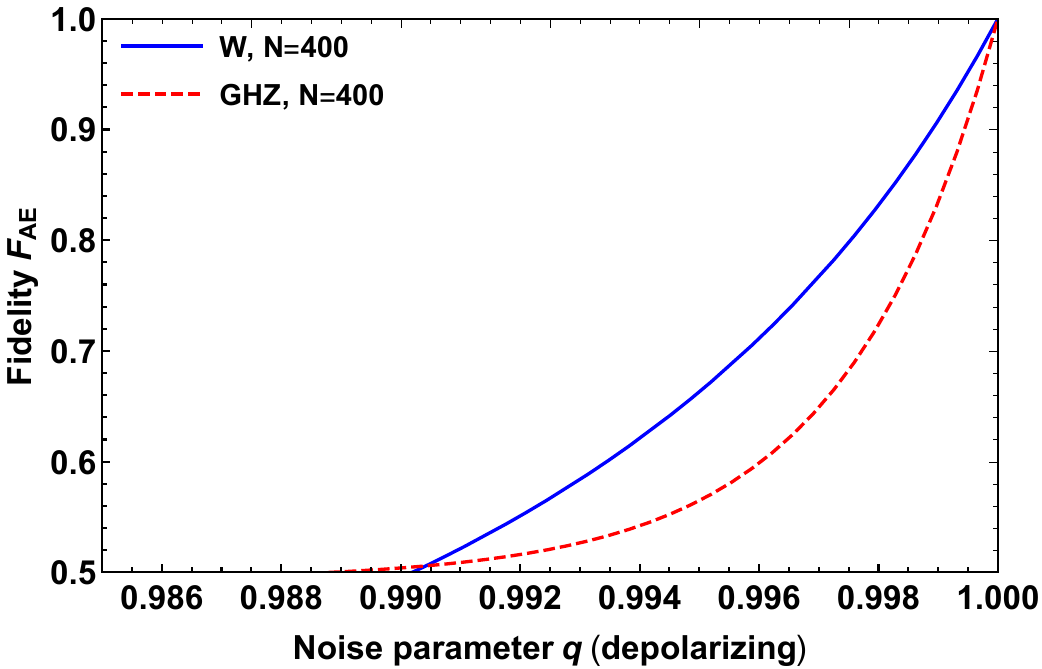}
\caption{Fidelity of anonymous entanglement for $N=400$.}\label{app:fig:fid400}   
\end{minipage}
\hspace{0.01\textwidth}
\begin{minipage}[t]{.45\textwidth} 
  \centering
  \includegraphics[width=\linewidth]{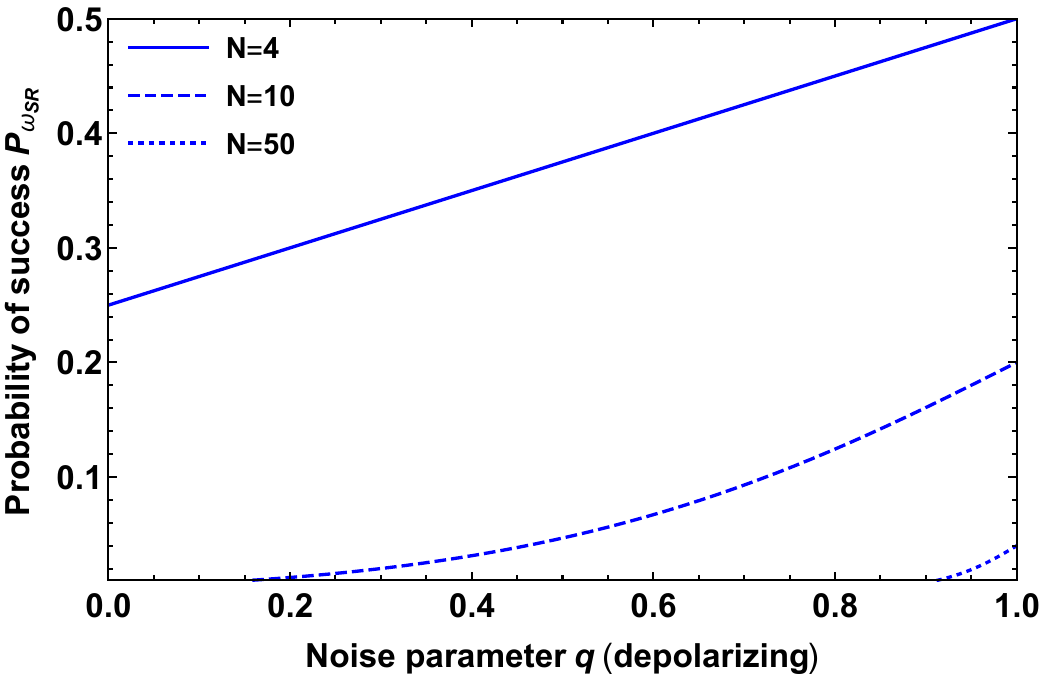}
\caption{Probability of success in Protocol 1 in the presence of the depolarizing noise, $N = \DE{4,10,50}$.}\label{app:fig:prob_succ}   
\end{minipage}
\end{figure}
As discussed, the noise parameter threshold $q^*$ for $N=182$ nodes becomes larger for the $W$ state: $q_W^* = 0.979057$, $q_{GHZ}^* = 0.979043$, $q_W^*>q_{GHZ}^*$. This means that for $N\geq 182$ the W state tolerates less noise than the GHZ; see Fig. \ref{app:fig:thresh_close}. However, we numerically see that there exists a value of $q>q^*_W$ for which $F_{AE}(\omega_{SR})>F_{AE}(\gamma_{SR})$. As an example for $N=400$ see Fig. \ref{app:fig:fid400}.

Moreover, we provide an analytical expression for the probability of success in our protocol, defined as $P_{\omega_{SR}} := \Tr[\Lambda^{\otimes N}(\ketbra{\tn{W}}_N) \cdot |\vec{0}\rangle\!\langle\vec{0}|_{N-2}]$,
which for the depolarizing noise assumes the form, 
\begin{align}
P_{\omega_{SR}} = \frac{(q+1)^{N-3}(N(1-q)+4q)}{N2^{N-2}}.
\end{align}
Examples of $P_{\omega_{SR}}$ as a function of $q$ for  $N=\DE{4,10,50}$ are plotted in Fig. \ref{app:fig:prob_succ}. Note that for the dephasing noise $P_{\omega_{SR}} = \frac{2}{N}$, {since the measurement basis is not affected by the $Z$ noise.} 

\vspace{1em}
\textbf{Particle loss.} In the case when one of the particles of the W state is lost and the state is subjected to the network noise, the fidelity of anonymous entanglement can be expressed as
\begin{align}
F_{AE}(\tilde{\omega}_{SR}) = \frac{(1+q)(N^2(q-1)^2-8q^2+4Nq(1+q))}{4N(N(1-q)+4q)}
\end{align}
for the depolarizing noise, and
\begin{align}
F_{AE}(\tilde{\omega}_{SR}) = \frac{N-1}{N}(1-2q(1-q)).
\end{align}
for the dephasing noise. In Fig. \ref{app:fig:fid_dephloss} we plot the examples of $F_{AE}$ for $N=\DE{4,10,50}$ when the initial W state is subjected to one particle loss and the dephasing noise.

\begin{figure}
\begin{minipage}[t]{.45\textwidth} 
  \centering
  \includegraphics[width=\linewidth]{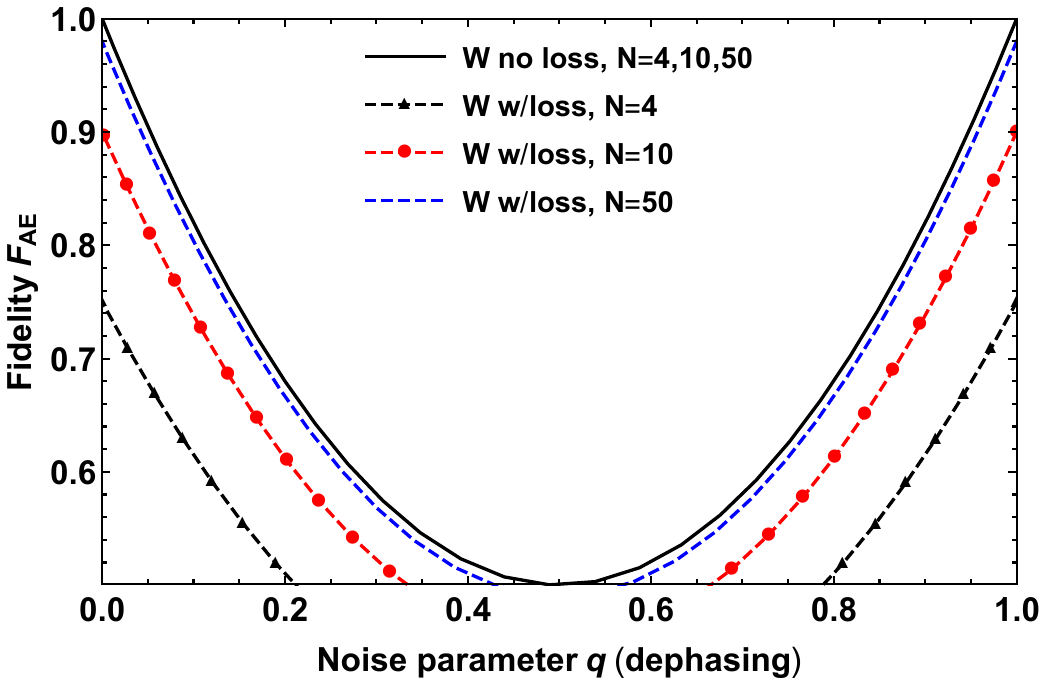}
\caption{Fidelity of anonymous entanglement for Protocol 1, as a function of the noise parameter for the dephasing channel in the presence of one particle loss.}\label{app:fig:fid_dephloss}   
\end{minipage}
\end{figure}

\textbf{Relay protocol.} Finally, in Tab. \ref{app:tab:relay} we present the values for anonymous entanglement in the relay protocol \citep{Yang2016} in the presence of the depolarizing noise.

\begin{table}[H]
\centering
  \caption{Fidelity of anonymous entanglement for the relay scheme \cite{Yang2016} in the $N$-fold noisy network for the depolarizing channel. Note that for the depolarizing parameter $q = 0.8$ the anonymous entanglement created between nodes 1 and 6 is not useful in the sense of Def. \ref{def:FAE_useful}.} \label{app:tab:relay}
  \begin{tabular}
      {c|C{3cm}|C{3cm}} 
      Scenario & $F_{AE}$ for $q = 0.8$ & $F_{AE}$ for $q = 0.95$ \\ \hline
      \begin{minipage}[4cm]{.3\textwidth}\centering
      \includegraphics[width=0.7\textwidth]{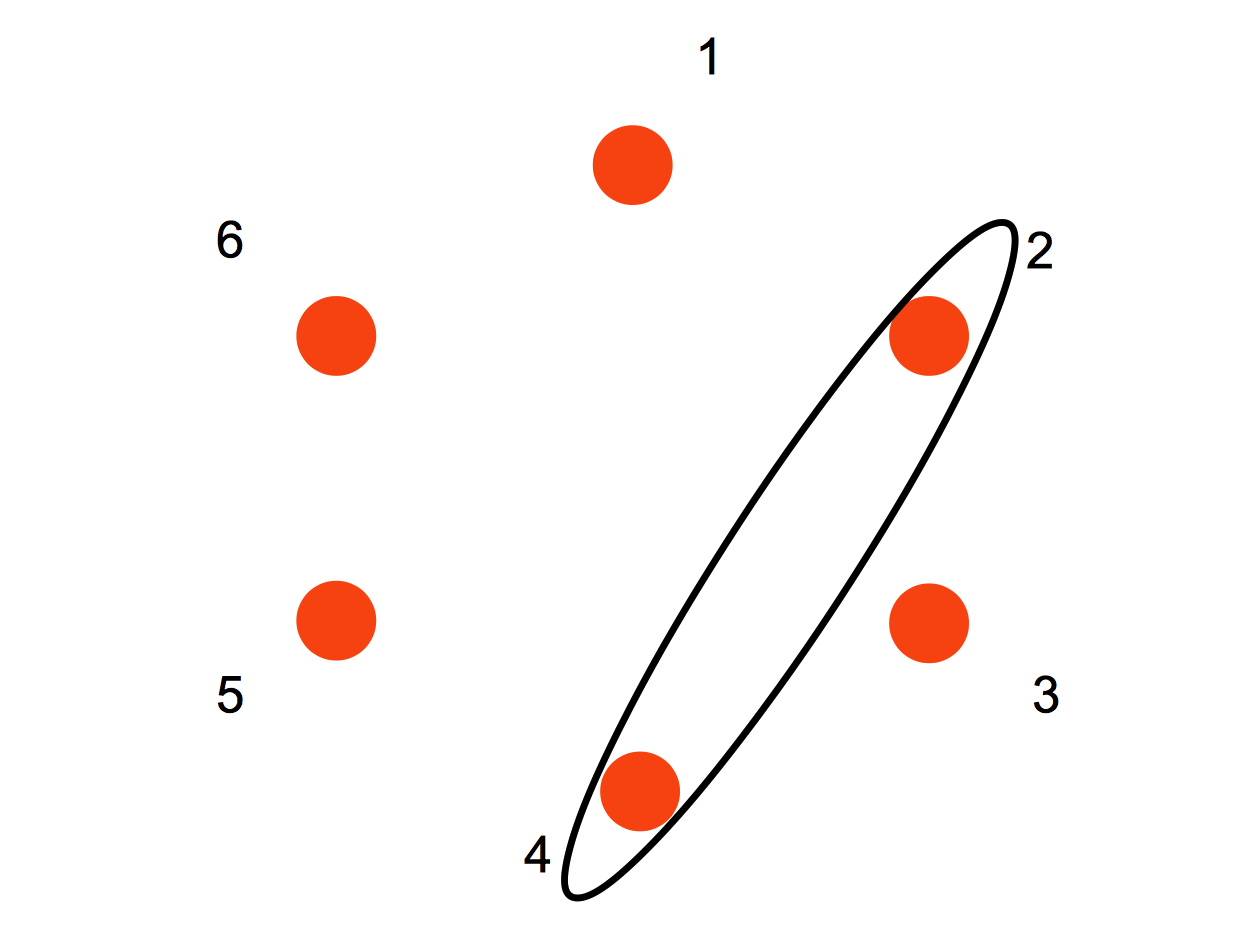}
      \\
     \end{minipage}
      & 0.5738 & 0.8625 \\ \hline
      \begin{minipage}[4cm]{.3\textwidth}\centering
      \includegraphics[width=0.7\textwidth]{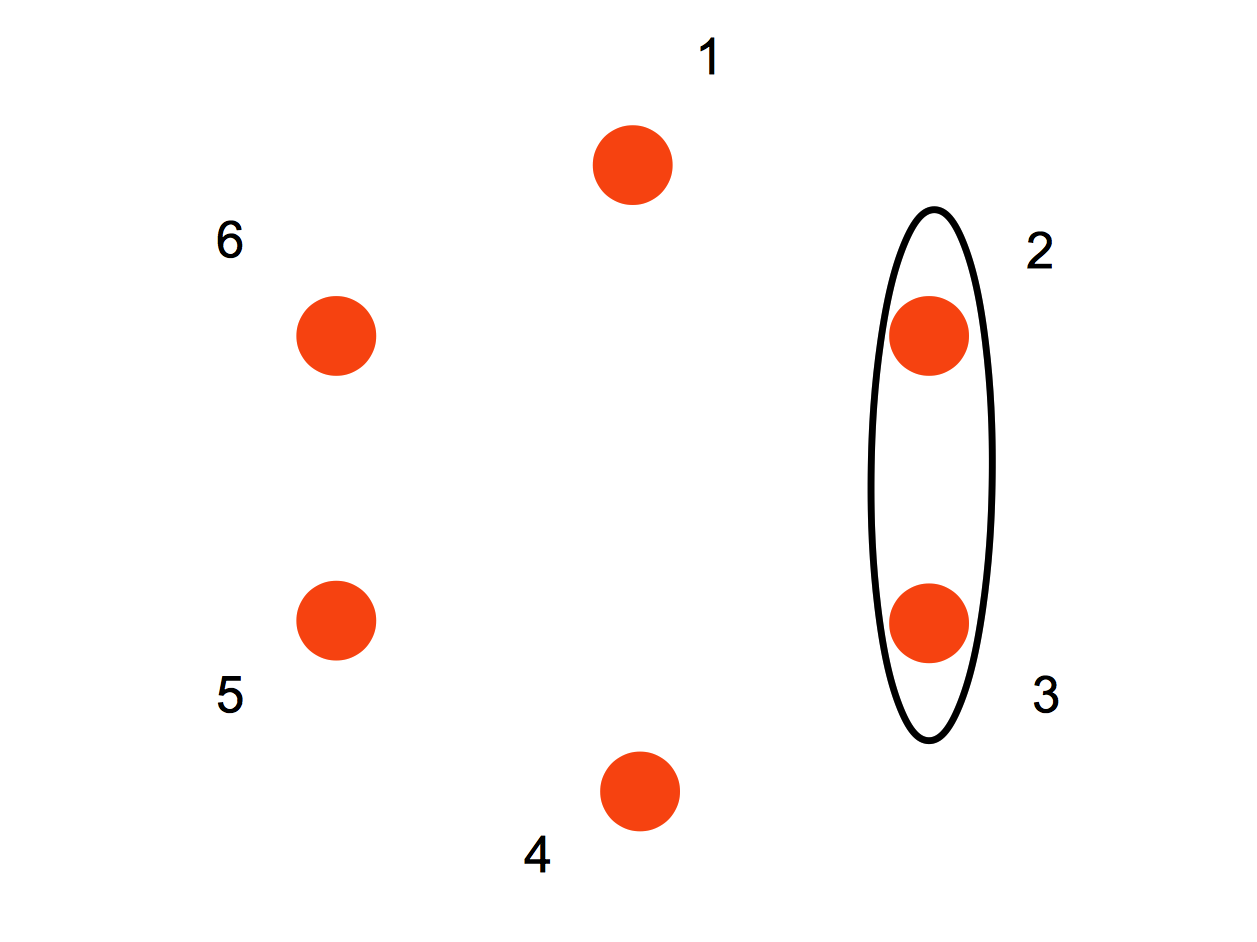}
     \end{minipage}
      & 0.6138 & 0.8744 \\  \hline
	  \begin{minipage}[4cm]{.3\textwidth}\centering
      \includegraphics[width=0.7\textwidth]{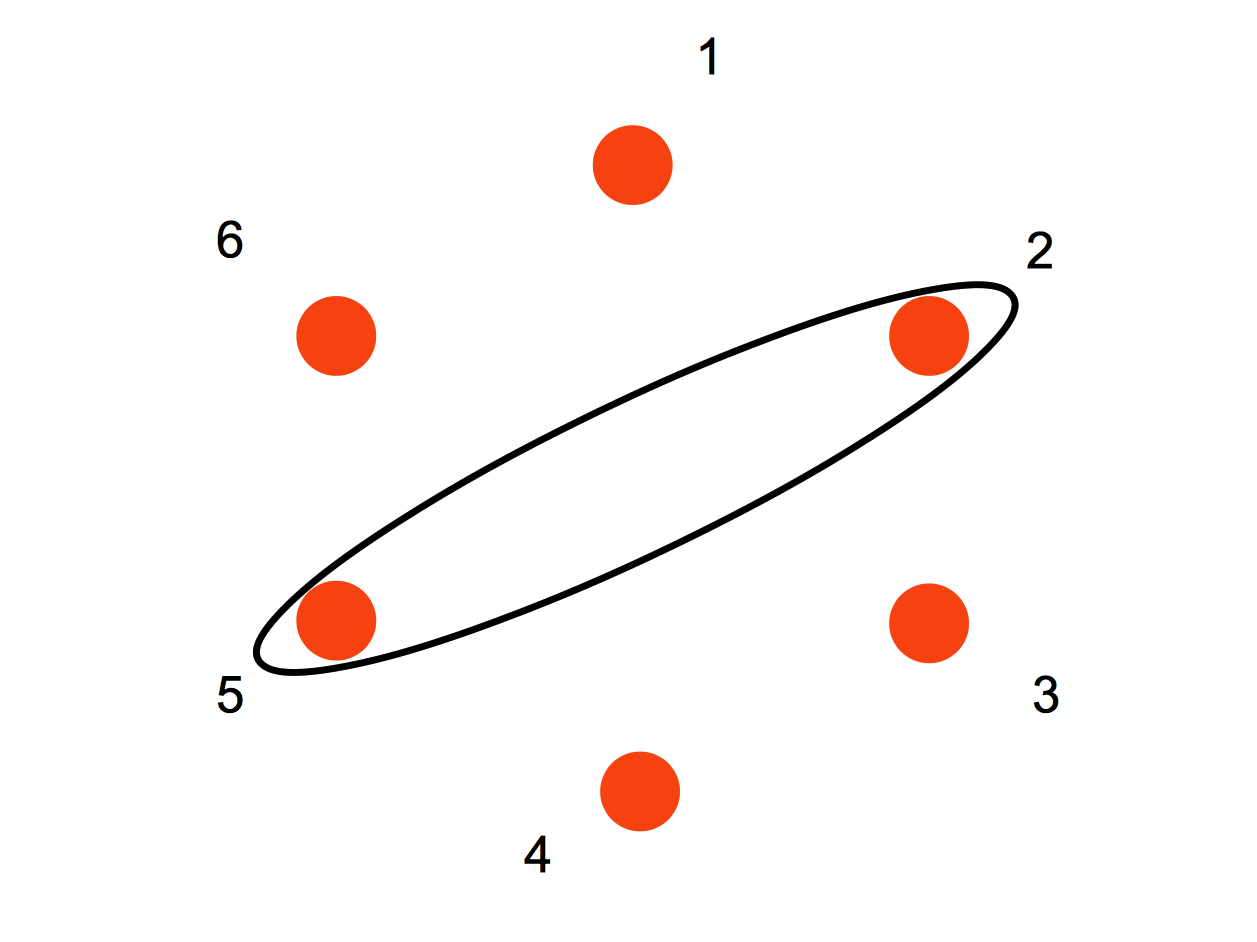}
     \end{minipage}
      & 0.5418 & 0.8512 \\ \hline
      \begin{minipage}[4cm]{.3\textwidth}\centering
      \includegraphics[width=0.7\textwidth]{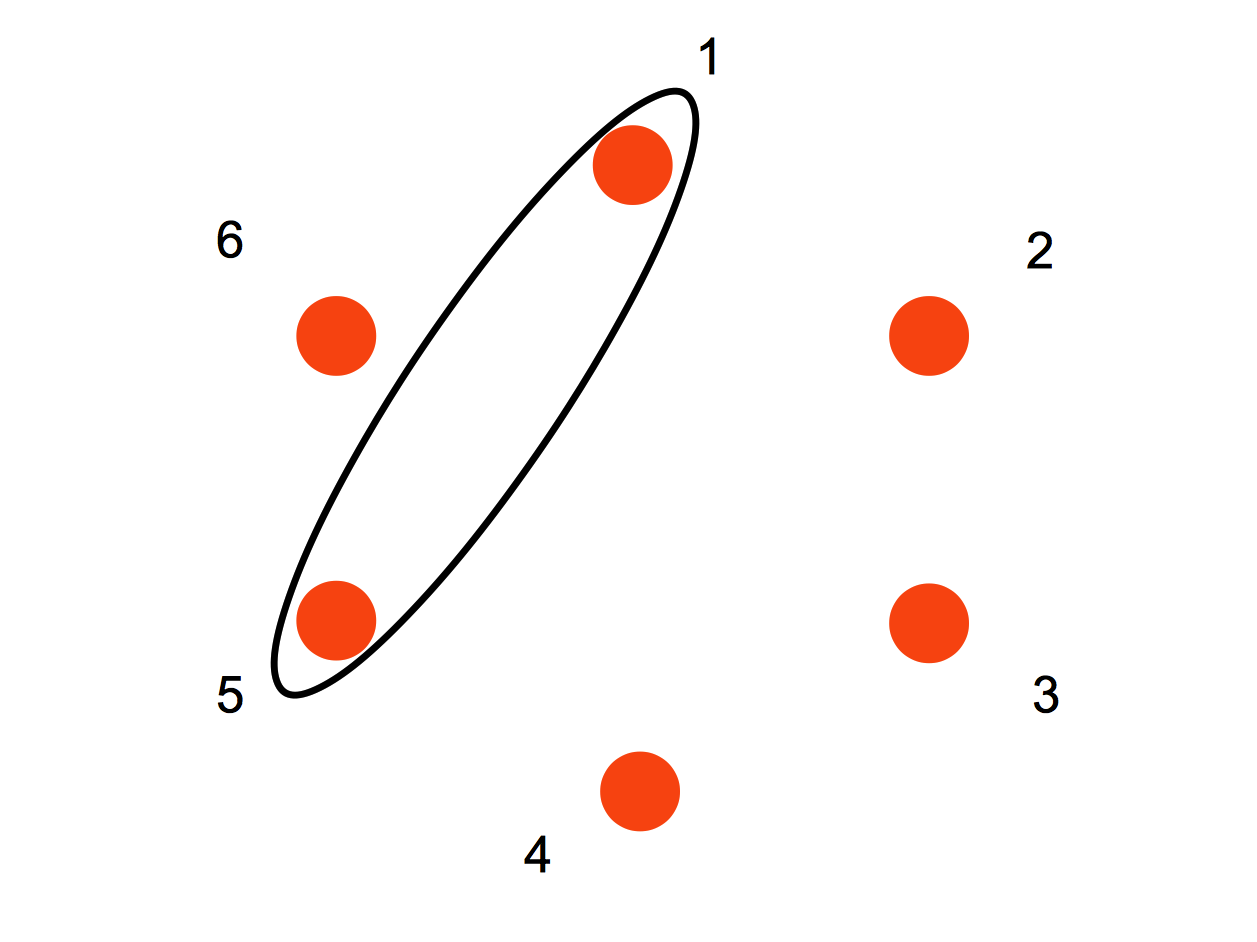}
     \end{minipage}
      & 0.5162 & 0.8405 \\ \hline
      \begin{minipage}[4cm]{.3\textwidth}\centering
      \includegraphics[width=0.7\textwidth]{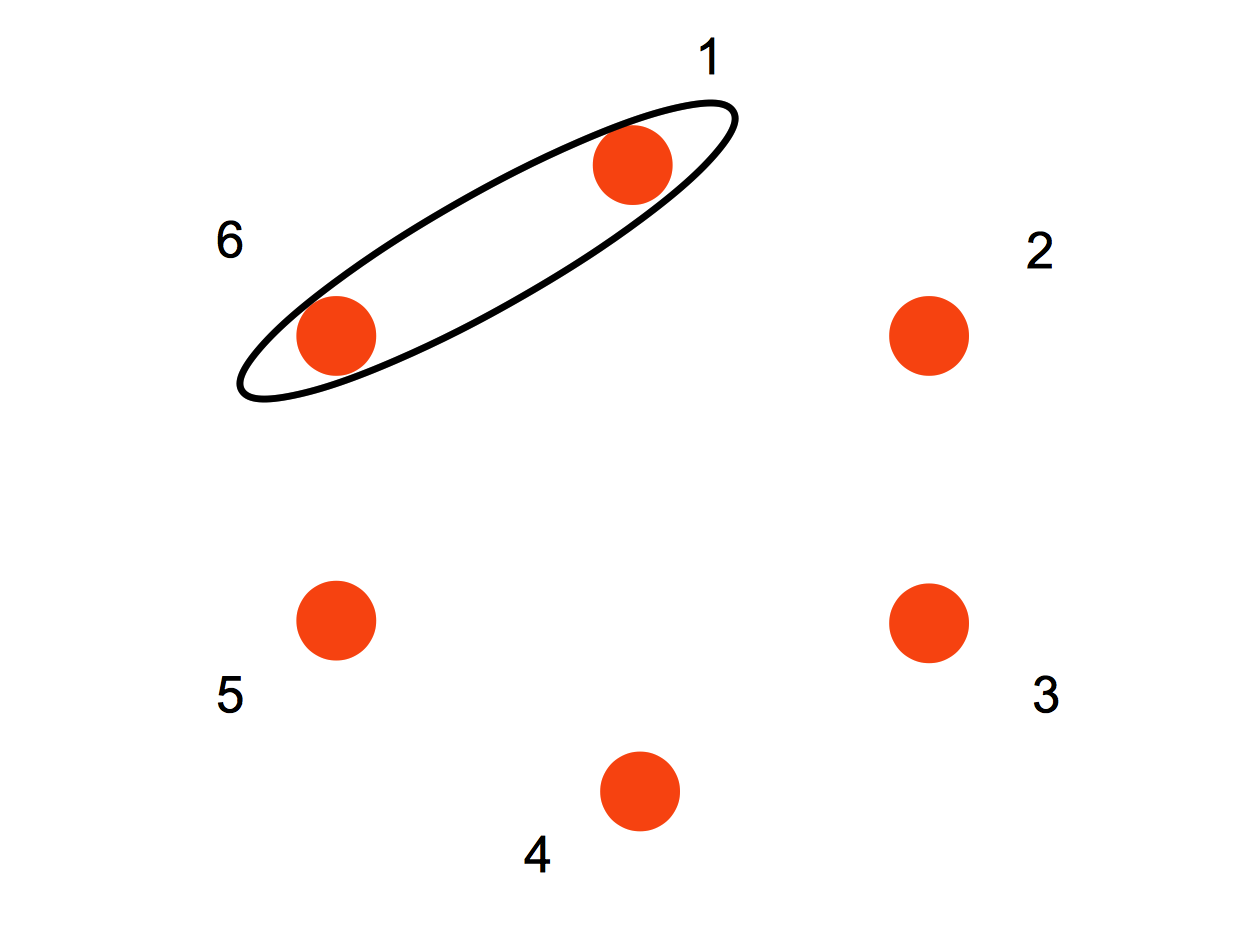}
     \end{minipage}
      & 0.4958 & 0.8303 \\ \hline
  \end{tabular}
\end{table}

\end{widetext}

\end{document}